\documentclass[a4paper,11pt]{article}
\usepackage[utf8]{inputenc}
\usepackage{hyperref}
\usepackage{fullpage}
\usepackage{graphicx}
\usepackage[small]{caption}
\usepackage{amsmath,amssymb,amsthm}
\allowdisplaybreaks

\usepackage[noend]{algpseudocode}
\usepackage{xspace}
\usepackage{nicefrac}
\usepackage{xparse}
\usepackage{bm}
\usepackage{xcolor}
\usepackage{pifont}
\usepackage{algorithm}

\newtheorem{thrm}{Theorem}

\newcommand{\cstar}{c_\star}
\newcommand{\Nin}[1]{N^{_-}_{#1}}
\newcommand{\Nout}[1]{N^{_+}_{#1}}

\newcommand{\PLTR}{\ensuremath{\text{PLTR}}\xspace}
\newcommand{\PLTRR}{\ensuremath{\text{R-PLTR}}\xspace}
\newcommand{\gplus}{g^{\text{\tiny{+}}}}
\newcommand{\gminus}{g^{\text{\small{-}}}}
\newcommand{\MOV}{\ensuremath{\text{MoV}}\xspace}
\newcommand{\MOVD}{\ensuremath{\text{MoV}_D}\xspace}
\newcommand{\emptyliveedge}{G'_\emptyset}

\newcommand{\SOLk}[1]{\text{SOL}_{DkS}(#1)}

\newcommand{\OPT}{\text{OPT}}
\newcommand{\OPTk}{\text{OPT}_{DkS}}

\DeclareMathOperator*{\argmax}{arg\,max}
\newcommand{\card}[1]{\vert #1 \vert}

\NewDocumentCommand\Ex{gg}{%
	\ensuremath{%
		\mathbf{E} \IfNoValueTF{#1}{}{\left[ #1 \IfNoValueTF{#2}{}{\cond #2} \right]}
	}
}

\RenewDocumentCommand\Pr{gg}{%
	\ensuremath{%
		\mathbf{P} \IfNoValueTF{#1}{}{\left( #1 \IfNoValueTF{#2}{}{\cond #2} \right)}
	}
}

\title{\textbf{Election Control through Social Influence\\with Unknown Preferences}} 
\author{Mohammad {Abouei Mehrizi} \\
		{\small{}Gran Sasso Science Institute}\\
		{\small{}L'Aquila, Italy}\\
		{\small{}\texttt{mohammad.aboueimehrizi@gssi.it}}\\
    \and Federico Cor\`o\\ 
        {\small{}Sapienza University of Rome}\\
		{\small{}Rome, Italy}\\
		{\small{}\texttt{federico.coro@uniroma1.it}}\\
	\and Emilio Cruciani \\
		{\small{}Inria, I3S Lab, UCA, CNRS}\\
		{\small{}Sophia Antipolis, France}\\
		{\small{}\texttt{emilio.cruciani@inria.fr}}\\
    \and Gianlorenzo {D'Angelo}\\
		{\small{}Gran Sasso Science Institute}\\
		{\small{}L'Aquila, Italy}\\
		{\small{}\texttt{gianlorenzo.dangelo@gssi.it}}\\
}
\date{}

\begin{document}

\maketitle
\begin{abstract}
The election control problem through social influence asks to find a set of nodes in a social network of voters to be the starters of a political campaign aiming at supporting a given target candidate. 
Voters reached by the campaign change their opinions on the candidates. 
The goal is to shape the diffusion of the campaign in such a way that the chances of victory of the target candidate are maximized.
Previous work shows that the problem can be approximated within a constant factor in several models of information diffusion and voting systems, assuming that the controller, i.e., the external agent that starts the campaign, has full knowledge of the preferences of voters. 
However this information is not always available since some voters might not reveal it.
Herein we relax this assumption by considering that each voter is associated with a probability distribution over the candidates. 
We propose two models in which, 
when an electoral campaign reaches a voter, this latter modifies its probability distribution according to the amount of influence it received from its neighbors in the network.
We then study the election control problem through social influence on the new models:
In the first model, under the Gap-ETH, election control cannot be approximated within a factor better than $1/n^{o(1)}$, where $n$ is the number of voters; in the second model, which is a slight relaxation of the first one, the problem admits a constant factor approximation algorithm.
\end{abstract}

\cleardoublepage
\section{Introduction}
Social media play a fundamental role in everyone's life providing information, entertainment, and learning.
Many social media users prefer to access social network platforms such as Facebook or Twitter before news websites as they provide faster means for information diffusion~\cite{Matsa2018}.
As a consequence, online social networks are also exploited as a tool to alter users' opinions. The extent to which the opinions of an individual are conditioned by social interactions is called social influence. It has been observed that social influence starting from a small set of individuals may generate a cascade effect that allows to reach a large part of the network.
Recently, this capability has been used to affect the outcome of political elections.
There exists evidence of political intervention which shows the effect of social media manipulation on the elections outcome, e.g., by spreading fake news~\cite{pennycook2018prior}. 
A real-life example is in the 2016 US election where a study 
showed that on average 92\% of people remembered pro-Trump fake news and 23\% of them remembered pro-Clinton fake news~\cite{allcott2017social}.
Several other cases have been studied~\cite{bond201261,ferrara2017disinformation,kreiss2016seizing,ElectionCampaigningSebastianStier2018}.

There exists a wide literature about manipulation of voting systems; we point the reader to a recent survey~\cite{faliszewski2016control}.
Despite that, only few studies focus on the problem of controlling the outcome of political elections through the spread of information in social networks.
The \emph{election control} problem~\cite{wilder2018controlling} consists in selecting a set of nodes of a network to be the starters of a diffusion with the aim of maximizing the chances for a target candidate to win an election. 
In particular, in the \emph{constructive} election control problem, the goal is to maximize the \emph{Margin of Victory} (\MOV) of the target candidate on its most critical opponent, i.e., the difference of votes (or score, depending on the voting system) between the two candidates after the effect of social influence.
A variation of the problem, known as \emph{destructive} election control, aims at making a target candidate lose.
Both problems have been originally analyzed under the Independent Cascade Model (ICM)~\cite{kempe2015maximizing}, and considering \emph{plurality voting}; approximation and hardness-of-approximation results are provided~\cite{wilder2018controlling}.
Cor\`o et al.~\cite{IJCAI-19,aamas19} analyzed the problem in arbitrary scoring rules voting systems under the Linear Threshold Model (LTM)~\cite{kempe2015maximizing}, providing constant factor approximation algorithms.
It has been later shown that it is \mbox{$\mathit{NP}$-hard} to find any constant factor approximation in the multi-winner scenario~\cite{sirocco2020}.

Faliszewski et al.~\cite{faliszewski2018opinion} examine bribery in an opinion diffusion process with voter clusters: each node is a cluster of voters, represented as a weight, with a specific list of candidates; there is an edge between two nodes if they differ by the ordering of a single pair of adjacent candidates. The authors show that making a specific candidate win in their model is \mbox{$\mathit{NP}$-hard} and fixed-parameter
tractable with respect to the number of candidates.
Bredereck et al.~\cite{bredereck2017manipulating} studied the problem of manipulating diffusion on social networks, though not specifically in the context of elections. They show that identifying successful manipulation via bribing, adding/deleting edges, or controlling the order of asynchronous updates are all computationally hard problems.
A similar approach is taken by Apt et al.~\cite{AptM14}, where the authors introduce a threshold model for social networks in order to characterize the role of social influence in the global adoption of a commercial product.

\paragraph{Contribution.}
In all previous works it is assumed that the controller knows the preference list of each voter. However, this assumption is not always satisfied in realistic scenarios as voters may not reveal their preferences to the controller.
Herein, in Section~\ref{ssec:problem}, we introduce two new models, \emph{Probabilistic Linear Threshold Ranking} (\PLTR) and \emph{Relaxed-\PLTR} (\PLTRR), that encompass scenarios where the preference lists of the voters are not fully revealed.
Specifically, we use an uncertain model in which the controller only knows, for each voter, a probability distribution over the candidates. In fact, in applied scenarios, the probability distribution could be inferred by analyzing previous social activity of the voters, e.g., re-tweets or likes of politically oriented posts.
We envision that some given focused news about a target candidate spread through the network as a message.
We model such a diffusion via the LTM~\cite{kempe2015maximizing}.
The message will have an impact on the opinions of voters who received it from their neighbors, leading to a potential change of their vote if the neighbors exercise a strong influence on them. 
With this intuition in mind, in our models, the probability distribution of the voters reached by the message is updated as a function of the degree of influence that the senders of the message have on them.
The rationale is that the controller, without knowing the exact preference list which is kept hidden, can just update its estimation on it by considering the mutual degree of influence among voters.
We acknowledge that our models do not cover all scenarios that can arise in election control, e.g., messages about multiple candidates. However they represent a first step towards modeling uncertainty.

We study on our models both the constructive and destructive election control problems.
We show in Section~\ref{sec:hardness} that the election control problem in $\PLTR$ is at least as hard to approximate as the Densest-$k$-Subgraph problem~\cite{manurangsi2017densestk}.
This result implies several conditional hardness of approximation bounds for our problem, for example it cannot be approximated within any constant factor, unless the Unique Game Conjecture holds and it cannot be approximated to within any polynomial factor if the Exponential Time Hypothesis holds.
However, these hardness of approximation bounds do not hold for the election control problem in $\PLTRR$, for which we can show that the problem remains $\mathit{NP}$-hard.

In Section~\ref{sec:approximation} we provide an algorithm that guarantees a constant factor approximation to the constructive and destructive election control problems in $\PLTRR$.
In the relaxed model, $\PLTRR$, also ``partially-influenced'' nodes change their probability distribution.
Although this simple modification is enough to make the problem substantially easier, preliminary experimental results
show that the hardness of approximation for $\PLTR$ is purely theoretical and is due to hard instances in the reduction.

In Section~\ref{sec:simulation} we present the simulation of our models and algorithm on two real-world datasets.

\section{Influence Models and Problem Statement}
\label{sec:bg}
\paragraph{Background.}
\label{ssec:influence}
Influence Maximization is the problem of finding a subset of the most influential users in a social network with the aim of maximizing the spread of information given a particular diffusion model.
In this work, we focus on the diffusion model known as \emph{Linear Threshold Model} (LTM)~\cite{kempe2015maximizing}.
Given a graph $G=(V, E)$, each edge $(u,v) \in E$ has a weight $b_{uv} \in [0,1]$, each node $v \in V$ has a threshold $t_v \in [0,1]$ sampled uniformly at random and independently from the others, and the sum of the weights of the incoming edges of $v$ is $\sum_{(u,v) \in E} b_{uv} \leq 1$. 
Each node can be either \emph{active} or \emph{inactive}. 
Let $A_0$ be a set of initially active nodes and $A_t$ be the set of nodes active at time $t$. 
A node $v$ becomes active if the sum of the incoming active weights at time $t-1$ is greater than or equal to its threshold $t_v$, i.e., $v \in A_t$ if and only if $v \in A_{t-1}$ or $\sum_{u \in A_{t-1} : (u,v) \in E} b_{uv} \geq t_v$.

The process terminates at the first time $\tilde{t}$ in which the set of active nodes would not change in the next round, i.e., $A_{\tilde{t}} = A_{\tilde{t}+1}$.
We define the eventual set of active nodes as $A := A_{\tilde{t}}$ and the expected size of $A$ as $\sigma(A_0)$.
Given a budget $B$, the influence maximization problem consists in finding a set of nodes $A_0$ of size $B$, called \emph{seeds}, in such a way that $\sigma(A_0)$ is maximum. 

Kempe et al.~\cite{kempe2015maximizing} showed that the distribution of active nodes $A$, for any set $A_0$, is equal to the distribution of the sets of nodes that are \emph{reachable} from $A_0$ in the set of random graphs called \emph{live-edge graphs}. A live-edge graph is a subgraph in which each node has at most one incoming edge.
Even if the number of live-edge graphs is exponential, by using standard Chernoff-Hoeffding bounds, it is possible to compute a $(1\pm\epsilon)$-approximation of $\sigma(A_0)$, for a given $A_0$, with high probability by sampling a polynomial number of live-edge graphs.
Moreover, $\sigma(A_0)$ is monotone and submodular w.r.t.\ to the initial set $A_0$; hence, an optimal solution can be approximated to a factor of $1-1/e$ using a simple greedy algorithm~\cite{NWF78}.
There has been intensive research on the problem in the last decade. We point the reader to a recent survey on the topic~\cite{li2018influence}.

\paragraph{Notation.}
\label{ssec:problem}

Let $G=(V, E)$ be a directed graph representing a social network of voters and their interactions.
We denote the set of $m$ candidates running for the election as $C=\{c_1, c_2, \ldots, c_m\}$ and the target candidate as $\cstar \in C$.
Each node $v \in V$ has a probability distribution over the candidates $\pi_v$, where $\pi_v(c_i)$ is the probability that $v$ votes for candidate $c_i$;
then for each $v \in V$ we have that $\pi_v(c_i) \geq 0$ for each candidate $c_i$ and $\sum_{i=1}^{m} \pi_v(c_i) = 1$.
Moreover, we denote by $\Nin{v}$ and $\Nout{v}$, respectively, the sets of incoming and outgoing neighbors for each node $v \in V$.
For each candidate $c_i$, we assume that $\pi_v(c_i)$ is at least a polynomial fraction of the number of voters, i.e., $\pi_v(c_i) = \Omega(1/|V|^\gamma)$ for some constant $\gamma > 0$.%
\footnote{\label{footnote_weights}The assumption is used in the approximation results, since Influence Maximization problem with exponential (or exponentially small) weights on nodes is an open problem. However, the assumption is realistic: Current techniques to estimate such parameters generate values linear in the number of messages shared by a node.}
Let $X_v(c_i)$ be an indicator random variable, where $X_v(c_i) = 1$ if $v$ votes for $c_i$, with probability $\pi_v(c_i)$, and $X_v(c_i) = 0$ otherwise.
We define the expected \emph{score} of a candidate $c_i$ as the expected number of votes that $c_i$ obtains from the voters
\(
F(c_i, \emptyset) := \Ex{\sum_{v \in V} X_v(c_i)} = \sum_{v \in V} \pi_v(c_i).
\)

\paragraph{PLTR Model.}
As in LTM, each node $v$ has a threshold $t_v \in [0, 1]$; each edge $(u, v) \in E$ has a weight $b_{uv}$, that models the influence of node $u$ on $v$, with the constraint that, for each node $v$, $\sum_{u : (u,v) \in E} b_{uv} \leq 1$. 
We assume the weight of each existing edge $(u,v)$ not to be too small, i.e., $b_{uv} = \Omega(1/|V|^\gamma)$ for some constant $\gamma > 0$.\footnotemark[4]%

Given an initial set of seed nodes $S$, the diffusion process proceeds as in LTM: Inactive nodes become active if the sum of the weights of incoming edges from active neighbors is greater than or equal to their threshold. 
Mainly, we are modeling the spread of some ads/news about the target candidate: Active nodes receive the message and spread it to their neighbors.
Moreover, in \PLTR, active nodes are influenced by the message, increasing their probability of voting for the target candidate. In particular, an active node $v$ increases the probability of voting for $\cstar$ by an amount equal to the sum of the weights of its edges incoming from other active nodes, i.e., it adds $\sum_{u\in A \cap \Nin{v}} b_{uv}$ to the initial probability $\pi_v(\cstar)$. Then it normalizes to maintain $\pi_v$ as a probability distribution.
Formally, for each node $v \in A$, where $A$ is the set of active nodes at the end of LTM, the preference list of $v$ is denoted as $\tilde{\pi}_v$ and it is equal to:
\begin{equation}
\label{eq:updatecstar-updateci}
\tilde{\pi}_v(\cstar) = 
\frac{\pi_v(\cstar) + \sum_{u\in A \cap \Nin{v}} b_{uv}}{1 + \sum_{u\in A \cap \Nin{v}} b_{uv}}\quad\text{and}\quad
\tilde{\pi}_v(c_i) = \frac{\pi_v(c_i)}{1 + \sum_{u\in A \cap \Nin{v}} b_{uv}},
\end{equation}
for each \(c_i \neq \cstar\).
All inactive nodes $v \in V \setminus A$ will have $\tilde{\pi}_v(c_i) = \pi_v(c_i)$ for all candidates, including $\cstar$.
As for the expected score before the process, we can compute the expected final score of a candidate $c_i$ as
\[
 F(c_i, S) := \Ex{ \sum_{v \in V} X_v(c_i, S) }
 = \sum_{v \in V} \tilde{\pi}_v(c_i),
\]%
where $X_v(c_i, S)$ is the indicator random variable after the process, i.e., $X_v(c_i, S) = 1$ if $v$ votes for $c_i$, with probability $\tilde{\pi}_v(c_i)$, and $X_v(c_i, S) = 0$ otherwise. 

Let us denote by $\mathcal{G}$ the set of all possible live-edge graphs sampled from $G$. 
We can also compute $F(c_i, S)$ by means of live-edge graphs used in the LTM model as
\begin{equation}\label{eq:scoreliveedge}
F(c_i, S) = \sum_{G' \in \mathcal{G}} F_{G'}(c_i, S) \cdot \Pr(G'),
\end{equation}
where $F_{G'}(c_i, S)$ is the score of $c_i$ in $G'\in \mathcal{G}$ and $\Pr(G')$ is the probability of sampling live-edge $G'$.
More precisely, for the target candidate we have
\begin{equation*}%\label{eq:Fcstar}
 F_{G'}(\cstar, S)\! =\!\!\!\!\!\sum_{v \in R_{G'}\!(S)}\!\!\!\!\!\!\!  \frac{\pi_v(\cstar)\! +\! \sum_{u\in R_{G'}(S) \cap \Nin{v}} b_{uv}}{1 + \sum_{u\in R_{G'}(S) \cap \Nin{v}} b_{uv}} +\!\!\!\! \sum_{v \in V\setminus R_{G'}\!(S)}\!\!\!\!\!\!\!\!\!\pi_v(\cstar),
\end{equation*}
where $R_{G'}(S)$ is the set of nodes reachable from $S$ in $G'$.
A similar formulation can be derived for $c_i \neq \cstar$.

\paragraph{\PLTRR Model.}
In the next section we prove that the election control problem in \PLTR is hard to approximate to within a polynomial fraction of the optimum (Theorem~\ref{theorem:apx-hard}). 
However, we show that a small relaxation of the model allows us to approximate it to within a constant factor.
In the relaxed model, that we call \emph{Relaxed Probabilistic Linear Threshold Ranking} (\PLTRR), the probability distribution of a node is updated if it has at least an active incoming neighbor (also if the node is not active itself).
More formally, every node $v\in V$ (and not just every node $v\in A$ as in \PLTR) changes its preference by updating its probability distribution via Eq.~\eqref{eq:updatecstar-updateci}; thus also nodes that have at least an active incoming neighbor can change.
The rationale is that a voter might slightly change its opinion about the target candidate if it receives some influence from its active incoming neighbors even if the received influence is not enough to activate it (thus making it propagate the information to its outgoing neighbors). 
Therefore, we include this small amount of influence in the objective function.
In the next section, we show that election control in \PLTRR is still \mbox{$\mathit{NP}$-hard}, and then we give an algorithm that guarantees a constant approximation ratio in this setting.

\paragraph{Problem Statement.}
In the \emph{constructive election control} problem we maximize the expected Margin of Victory (\MOV) of the target candidate w.r.t.\ its most voted opponent, akin to~\cite{IJCAI-19,wilder2018controlling}.
We define the $\MOV(S)$ obtained starting from $S$ as the expected increase, w.r.t.\ the value before the process, of the difference between the score of $\cstar$ and that of the most voted opponent.%
\footnote{\label{footnote_mov}The increment in margin of victory, instead of just the margin, cannot be negative and gives well defined approximation ratios.}
Formally, if $c$ and $\hat{c}$ are respectively the candidates different from $\cstar$ with the highest score before and after the diffusion process
{\begin{equation}\label{def:mov}
    \MOV(S) := F(c, \emptyset) - F(\cstar, \emptyset) - \left( F(\hat{c}, S) - F(\cstar, S) \right).
\end{equation}}%

Given a budget $B$, the \emph{constructive election control problem} asks to find a set of seed nodes $S$, of size at most $B$, that maximizes $\MOV(S)$.
It is worth noting that \MOV can also be expressed as a function of the score gained by candidate $\cstar$ and the score lost by its most voted opponent $\hat{c}$ at the end of the process.
We define the score gained and lost by a candidate $c_i$ as 
{\begin{gather*}
\gplus(c_i, S) := F(c_i, S) - F(c_i, \emptyset), 
%\\
\gminus(c_i, S) := F(c_i, \emptyset) - F(c_i, S).
\end{gather*}}%
Therefore, we can rewrite $\MOV(S)$ as
{\begin{equation}
    \MOV(S) = \gplus(\cstar, S) + \gminus(\hat{c}, S) - F(\hat{c}, \emptyset) + F(c, \emptyset).
\label{eq:mov}
\end{equation}}%

The \emph{destructive election control problem}, instead, aims at making the target candidate lose by minimizing its \MOV.
In this dual scenario, the probability distributions of the voters are updated slightly differently in our models, i.e., influenced voters have a lower probability of voting for the target candidate $\cstar$ mimicking the spread of ``negative'' news about $\cstar$.

\paragraph{Influencing Voters About Other Candidates.}

In our model the controller can send to the seed nodes a message in support of only one single candidate, e.g., latest news about the candidate.
We prove that the best strategy is that of sending messages in support of the target candidate $\cstar$, i.e., if the controller wants $\cstar$ to win, then, according to our models, the direct strategy of targeting voters with news about $\cstar$ is more effective than the alternative strategy of distracting the same voters with news about other candidates.

Indeed, it is not always sufficient to maximize the score of the target candidate to ensure his victory or to maximize the margin of victory, and it is  easy to find counter-examples of this strategy. Moreover, in the models of Wilder et al.~\cite{wilder2018controlling} and Cor\`o et al.~\cite{IJCAI-19} it could be convenient to increase the score of a third candidate in order to make the most voted opponent w.r.t.\ $\cstar$ lose score and favor $\cstar$.

However, as previously claimed, in our models this does not hold. In fact, we can distinguish between the three possible strategies:
\begin{itemize}
    \item $\MOV_1$: Influencing voters about $\cstar$.
    \item $\MOV_2$: Influencing voters about $\hat{c}$, i.e., the most voted opponent w.r.t.\ $\cstar$ at the end of the process.
    \item $\MOV_3$: Influencing voters about any other candidate $c$.
\end{itemize}
Let us now analyze the \MOV of $\cstar$ in these three different cases.
As described in Equation~\eqref{eq:mov}, a general formulation for \MOV is the following
{\begin{align*}
\MOV(S) &:= \gplus(\cstar, S) + \gminus(\hat{c}, S) + \Delta\\
 & = F(\cstar, S) - F(\cstar, \emptyset) + F(\hat c, \emptyset) - F(\hat c, S) + \Delta,
\end{align*}}%
where $S$ is the initial set of seed nodes and $\Delta$ is the sum of constant terms that are not modified by the process.
With some algebra, it is possible to compute the \MOV of $\cstar$ in such scenarios, getting the following formulations:
{\begin{align*}
\bullet\, \MOV_{1}(S) 
&=\sum_{v \in A} \frac{\left(1 + \pi_v(\hat{c}) - \pi_v(\cstar) \right) \sum_{u \in A \cap \Nin{v}} b_{uv}}{1 + \sum_{u \in A \cap \Nin{v}} b_{uv}} + \Delta;
\\
\bullet\, \MOV_{2}(S) 
&=\sum_{v \in A} \frac{\left(\pi_v(\hat{c}) - \pi_v(\cstar) -1 \right) \sum_{u \in A \cap \Nin{v}} b_{uv}}{1 + \sum_{u \in A \cap \Nin{v}} b_{uv}} + \Delta;
\\
\bullet\, \MOV_{3}(S) 
&=\sum_{v \in A} \frac{\left(\pi_v(\hat{c}) - \pi_v(\cstar)\right) \sum_{u \in A \cap \Nin{v}} b_{uv}}{1 + \sum_{u \in A \cap \Nin{v}} b_{uv}} + \Delta.
\end{align*}}%
We just need to observe that $\MOV_{1}(S) \geq \MOV_{2}(S)$ and that $\MOV_{1}(S) \geq \MOV_{3}(S)$ to conclude that it is always convenient to influence the voters about the target candidate whenever you want to maximize the \MOV of $\cstar$.
Therefore, in the remainder of the paper, we only focus on changing the score of the target candidate $\cstar$.
Note that the observations above hold both for $\PLTR$ and $\PLTRR$.

\section{Hardness Results}
\label{sec:hardness}

In this section we provide two hardness results related to election control in $\PLTR$ and $\PLTRR$.
In Theorem~\ref{theorem:apx-hard} we show that maximizing the \MOV in $\PLTR$ is at least as hard to approximate as the Densest-$k$-subgraph problem. This implies several conditional hardness of approximation bounds for the election control problem. Indeed, it has been shown that the Densest-$k$-subgraph problem is hard to approximate: to within any constant bound under the Unique Games with Small Set Expansion conjecture~\cite{raghavendra2010graph}; to within $n^{-1/(\log \log n)^c}$, for some constant $c$, under the exponential time hypothesis (ETH)~\cite{manurangsi2017densestk}; to $n^{-f(n)}$ for any function $f\in o(1)$, under the Gap-ETH assumption~\cite{manurangsi2017densestk}.
Then, in Theorem~\ref{theorem:np-hardness}, we show that maximizing the \MOV in $\PLTRR$ is still $\mathit{NP}$-hard.

\begin{thrm}
\label{theorem:apx-hard}
An $\alpha$-approximation algorithm to the election control problem in $\PLTR$ gives an $\alpha\beta$-approximation to the Densest $k$-Subgraph problem, for a positive constant $\beta<1$.
\end{thrm}
\begin{proof}
Given an undirected graph $G=(V,E)$ and an integer $k$, Densest $k$-Subgraph (DkS) is the problem of finding the subgraph induced by a subset of $V$ of size $k$ with the highest number of edges given that $k$ is fixed.

The reduction works as follows:
Consider the $\PLTR$ problem on $G$, where each undirected edge $\{u,v\}$ is replaced with two directed edges $(u,v)$ and $(v,u)$. Let us consider $m$ candidates and assume that all nodes initially have null probability of voting for all the candidates but one, different from $\cstar$, that we denote as $\hat{c}$. Formally we have that, $\pi_v(\hat{c}) = 1$ and $\pi_v(c_i) = \pi_v(\cstar) = 0$ for each $c_i \neq \hat{c}$ and for each $v \in V$. 
Assign to each edge $(u,v)\in E$ a weight $b_{uv} = \frac{1}{n^{\gamma}}$, for any fixed constant $\gamma\geq 4$ and $n = |V|$.

We show the reduction considering the problem of maximizing the score, because in the instance considered in the reduction the \MOV is exactly equal to twice the score.
In fact, the score of $\hat{c}$ after $\PLTR$ starting from any initial set $S$ is
\begin{align*}
&F(\hat{c}, S) 
= \sum_{v \in V} \tilde{\pi}_v(\hat{c})
= \sum_{v \in V \setminus A} \pi_v(\hat{c}) + \sum_{v \in A} \tilde{\pi}_v(\hat{c})
\\
&= |V| - |A| + \sum_{v \in A} \frac{1}{1+\sum_{u \in A \cap \Nin{v}} \frac{1}{n^{\gamma}}}
\\
&= |V| - \sum_{v \in A} \left( 1 - \frac{1}{1+\sum_{u \in A \cap \Nin{v}} \frac{1}{n^{\gamma}}} \right)
\\
&= |V| - \sum_{v \in A} \left( \frac{\sum_{u \in A \cap \Nin{v}} \frac{1}{n^{\gamma}}}{1+\sum_{u \in A \cap \Nin{v}} \frac{1}{n^{\gamma}}} \right)
= |V| - F(\cstar, S),
\end{align*}
because $( {\sum_{u \in A \cap \Nin{v}} \frac{1}{n^{\gamma}}} ) / ( {1+\sum_{u \in A \cap \Nin{v}} \frac{1}{n^{\gamma}}} ) = \tilde{\pi}_v(\cstar)$ and $\pi_v(\cstar) = 0$ for each $v \in V$.
Thus, according to the definition of \MOV in Equation~(6), we have that 
\[
\MOV(S) = |V| - (|V| - F(\cstar,  S) - F(\cstar,  S)) = 2F(\cstar, S).
\]

To compute the expected final score of the target candidate we average its score in all live live-edge graph in $\mathcal{G}$, according to Formula~(3).
In our reduction, the empty live-edge graph $\emptyliveedge = (V, \emptyset)$ is sampled \emph{with high probability}, i.e., with probability at least $1- \frac{1}{n^{\gamma -2}}$:
\begin{align*}
\Pr{\emptyliveedge}
&= \prod_{v \in V} \left(1 - \sum_{u \in \Nin{v}}b_{uv}\right)
= \prod_{v \in V} \left(1 - \frac{\card{\Nin{v}}}{n^{\gamma}}\right) \\
&\geq \prod_{v \in V} \left(1 - \frac{1}{n^{\gamma-1}} \right)
= \left( 1-\frac{1}{n^{\gamma-1}} \right)^n
\\ 
&\stackrel{(a)}{=} \sum_{i=0}^n \binom{n}{i} (1)^{n-i}\left(\frac{-1}{n^{\gamma-1}}\right)^i = 
\sum_{i=0}^n \binom{n}{i} \frac{(-1)^i}{n^{i(\gamma-1)}}
\\
&\stackrel{(b)}{\geq}
\binom{n}{0} - \binom{n}{1}\frac{1}{n^{\gamma -1 }} + \sum_{i=2}^{\lfloor n/2\rfloor} \left(\binom{n}{i} \frac{1}{n^{2i(\gamma-1)}} - \binom{n}{i+1} \frac{1}{n^{(2i+1)(\gamma-1)}}\right)
\\
&\stackrel{(c)}{\geq} 1- \frac{1}{n^{\gamma-2}}
\end{align*}
where in $(a)$ we used the binomial expansion, $(b)$ is due to last negative term in the lhs that does not appear in the rhs when $n$ is even, and $(c)$ is due to 
\[ 
 \binom{n}{i} \frac{1}{n^{2i(\gamma-1)}} \geq \binom{n}{i+1} \frac{1}{n^{(2i+1)(\gamma-1)}},
\]
for any $\gamma \geq 2$.
Since $\Pr{\emptyliveedge} \leq 1$, then $\Pr{\emptyliveedge} = \Theta(1)$.  Moreover, $\sum_{G' \neq \emptyliveedge} \Pr{G'} = \mathcal{O}\left( \frac{1}{n^{\gamma-2}} \right)$

The score obtained by $\cstar$ in a live-edge graph $G'$ starting from any initial set of seed nodes $S$ is

\begin{align*}
F_{G'}(\cstar, S)
&= \sum_{v \in R_{G'}(S)} \frac{\pi_v(\cstar) + \sum_{u \in R_{G'}(S) \cap \Nin{v}} \frac{1}{n^{\gamma}} }{ 1+\sum_{u \in R_{G'}(S) \cap \Nin{v}} \frac{1}{n^{\gamma}} } \\
&= \Theta \left(
\frac{1}{n^{\gamma}} \sum_{v \in R_{G'}(S)}\card{R_{G'}(S) \cap \Nin{v}}
\right),
\end{align*}

since $1 \leq 1+ \sum_{u \in R_{G'}(S) \cap \Nin{v}} \frac{1}{n^{\gamma}} \leq 2$ for each $v \in R_{G'}(S)$.
Also, note that $\sum_{v \in R_{G'}(S)}\card{R_{G'}(S) \cap \Nin{v}}$ is equal to the number of edges of the subgraph induced by the set $R_{G'}(S)$ of nodes reachable from $S$ in $G'$, which is not greater than $n^2$, and thus
\(
F_{G'}(\cstar, S) = \mathcal{O} \left( \frac{1}{n^{\gamma-2}} \right).
\)

Note that in the empty live edge graph $\emptyliveedge$ the set  $R_{\emptyliveedge}(S)$ at the end of LTM is equal to $S$, since the graph has no edges. 
Thus
\[ 
F_{\emptyliveedge}(\cstar, S)
= \frac{1}{n^{\gamma}} \cdot
\sum_{v \in S} \frac{ \card{S \cap \Nin{v}} }{ 1+\sum_{u \in S \cap \Nin{v}} \frac{1}{n^{\gamma}} }
\]
and since the denominator is, again, bounded by two constants we have that
\[ 
F_{\emptyliveedge}(\cstar, S)
= \Theta\left( \frac{ \sum_{v \in S}\card{S \cap \Nin{v}} }{n^{\gamma}} \right)
= \Theta\left( \frac{\SOLk{S}}{n^{\gamma}} \right),
\]
where $\SOLk{S} := \sum_{v \in S} \card{S \cap \Nin{v}}$ is the number of edges of the subgraph induced by $S$, i.e., the value of the objective function of DkS for solution $S$.

Thus, the expected final score of the target candidate is
\[
F(\cstar,S)
= \sum_{G' \in \mathcal{G}} F_{G'}(\cstar, S) \cdot \Pr(G') 
= F_{\emptyliveedge}(\cstar, S)\cdot \Pr(\emptyliveedge)
+ \sum_{G' \neq \emptyliveedge} F_{G'}(\cstar, S)\cdot \Pr(G').
\]
Since $F_{G'}(\cstar, S)$ and $\sum_{G' \neq \emptyliveedge} \Pr{G'}$ are in $\mathcal{O}\left( \frac{1}{n^{\gamma-2}} \right)$, then

\begin{align*}
\sum_{G' \neq \emptyliveedge} F_{G'}(\cstar, S) \cdot \Pr(G')
&= \mathcal{O}\left(\frac{1}{n^{\gamma-2}}\right) \sum_{G' \neq \emptyliveedge} \Pr(G') \\
&= \mathcal{O}\left( \frac{1}{n^{2(\gamma-2)}}\right)
= \mathcal{O}\left(\frac{\SOLk{S}}{n^{\gamma}}\right),
\end{align*}

for any $\gamma\geq4$.
Thus
\[ 
F(\cstar,S) = \Theta\left(
\frac{\SOLk{S}}{n^{\gamma}}
\right) \cdot \Theta(1)
+ \mathcal{O}\left(\frac{\SOLk{S}}{n^{\gamma}}\right)
\]
which means that
\(
F(\cstar,S) = \Theta\left(\frac{\SOLk{S}}{n^{\gamma}}\right).
\)
We apply the Bachmann-Landau definition of $\Theta$ notation: There exist three positive constants $n_0, \beta_1,$ and  $\beta_2$ such that, for all $n> n_0$,
\[ 
\beta_1 \frac{\SOLk{S}}{n^{\gamma}}
\leq
F(\cstar,S)
\leq
\beta_2 \frac{\SOLk{S}}{n^{\gamma}}.
\]
Note that, in this case, the constants $n_0$, $\beta_1$, and $\beta_2$ do not depend on the specific instance.

Since the previous bounds hold for any set $S$ we also have that $\beta_1 \frac{\OPTk}{n^{\gamma}} \leq \OPT \leq \beta_2 \frac{\OPTk}{n^{\gamma}}$, where $\OPT$ is the value of an optimal solution for $\PLTR$ and $\OPTk$ is the value of an optimal solution for DkS.

Suppose there exists an $\alpha$-approximation algorithm for $\PLTR$, i.e., an algorithm that finds a set $S$ s.t. the value of its solution is $\MOV(S) = 2F(\cstar,S) \geq \alpha \cdot \OPT$.
Then,
\[ 
\frac{\alpha}{2} \cdot \beta_1 \frac{\OPTk}{n^{\gamma}}
\leq
\frac{\alpha}{2} \cdot \OPT
\leq
F(\cstar,S)
\leq
\beta_2 \frac{\SOLk{S}}{n^{\gamma}}.
\]
Thus
\( \SOLk{S} \geq \frac{\alpha}{2}\frac{\beta_1}{\beta_2}\OPTk \),
i.e., it is an $\frac{\alpha\beta_1}{2\beta_2}$-approximation to DkS.
\end{proof}
As a corollary of Theorem~\ref{theorem:apx-hard} we get the conditional hardness of approximation bounds stated at the beginning of this section.

\begin{thrm}
Election control in $\PLTRR$ is $\mathit{NP}$-hard.
\label{theorem:np-hardness}
\end{thrm}

\begin{proof}
We prove the hardness by reduction from Influence Maximization under LTM, which is known to be $\mathit{NP}$-hard~\cite{kempe2015maximizing}.

Consider an instance $\mathcal{I}_{\text{LTM}} = (G, B)$ of Influence Maximization under LTM. 
$\mathcal{I}_{\text{LTM}}$ is defined by a weighted graph $G = (V, E, w)$ with weight function $w : E \rightarrow [0, 1]$ and by a budget $B$.
Let $\mathcal{I}_{\PLTRR} := (G', B)$ be the instance that corresponds to $\mathcal{I}_{\text{LTM}}$ on $\PLTRR$, defined by the same budget $B$ and by a graph $G' = (V', E', w')$ that can be built as follows:
\begin{enumerate}
\item Duplicate each vertex in the graph, i.e., we define the new set of nodes as $V' := V \cup \{v_{|V| + 1}, \ldots, v_{2|V|} \}$.
\item Add an edge between each vertex $v \in V$ to its copy in $V'$, i.e., we define the new set of edges as $E' := E \cup \{ (v_1, v_{|V| + 1}), \ldots, (v_{|V|}, v_{2|V|}) \}$.
\item Keep the same weight for each edge in $E$ and we set the weights of all new edges to $1$, i.e., $w'(e) = w(e)$ for each $e \in E$ and $w'(e) = 1$ for each $e \in E' \setminus E$. 
    Note that the constraint on incoming weights required by LTM is not violated by $w'$.
\item Consider $m$ candidates $\cstar, c_1, \ldots, c_{m-1}$. For each $v \in V$ we set $\pi_v(\cstar) = 1$ and $\pi_v(c_i) = 0$ for any other candidate $i \in \{1,\ldots,m-1\}$.
For each $v \in V' \setminus V$ we set $\pi_v(\cstar) = 0$, $\pi_v(c_1) = 1$ and $\pi_v(c_i) = 0$ for any other candidate $i \in \{2,\ldots,m-1\}$.
\end{enumerate}

Let $S$ be the initial set of seed nodes of size $B$ that maximizes $\mathcal{I}_{\text{LTM}}$ and let $A$ be the set of active nodes at the end of the process.
The value of the \MOV obtained by $S$ in $\mathcal{I}_{\PLTRR}$ is
$\MOV(S) = |V| - |V \setminus A|$.
Indeed, each node $v \in V$ in $G'$ has $\tilde{\pi}_v(\cstar) = \pi_v(\cstar) = 1$, because the probability of voting for the target candidate remains the same after the normalization.
Moreover, each node $v_i \in V\cap A$ influences its duplicate $v_{|V| + i}$ with probability 1 and therefore $\tilde{\pi}_{v_{|V| + i}}(\cstar) = (\pi_{v_{|V| + i}}(\cstar) + 1)/2 = \frac{1}{2}$.
Therefore, $F(\cstar, \emptyset) = F(c_1, \emptyset) = |V|$, 
$F(\cstar, S) = |V| + \frac{1}{2}|A|$, and $F(c_1, S) = |V\setminus A| + \frac{1}{2}|A|$.

Let $S$ be the initial set of seed nodes of size $B$ that achieves the maximum in $\mathcal{I}_{\PLTRR}$. Without loss of generality, we can assume that $S\subseteq V$, since we can replace any seed node $v_{|V|+i}$ in $V'\setminus V$ with its corresponding node $v_i$ in $V$ without decreasing the objective function. If $A$ is the set of active nodes at the end of the process, then by using similar arguments as before, we can prove that $\MOV(S) = |V| - |V \setminus A|$. 
Let us assume that $S$ does not maximize $\mathcal{I}_{\text{LTM}}$, then, $S$ would also not maximize $\mathcal{I}_{\PLTRR}$, which is a contradiction since $S$ is an optimal solution for $\mathcal{I}_{\PLTRR}$.

We can prove the $\mathit{NP}$-hardness for the case of maximizing the score by using the same arguments. In fact, notice that maximizing the score of $\cstar$, i.e., $F(\cstar, S) = |V| + \frac{1}{2}|A|$, is exactly equivalent to maximize the cardinality of the active nodes in LTM.
\end{proof}

\section{Approximation Results}
\label{sec:approximation}

In this section, we first show that we can approximate the optimal $\MOV$ to within a constant factor by optimizing the increment in the score of $\cstar$.
In detail we show that, given two solutions $S^*$ and $S^{**}$ such that $\gplus(\cstar,S^*)$ and $\MOV(S^{**})$ are maximum, then $\MOV(S^{*})\geq \frac{1}{3} \MOV(S^{**})$. 
Indeed, we show a more general statement that is: If a solution $S$ approximates $\gplus(\cstar,S^{*})$ within a factor $\alpha$, then $\MOV(S)\geq \frac{\alpha}{3}\MOV(S^{**})$.

Then we show that a simple greedy hill-climbing approach (Algorithm~\ref{alg:greedy}) gives a constant factor approximation to the problem of maximizing $\gplus(\cstar,S)$, where the constant is $\frac{1}{2}(1-\frac{1}{e})$. By combining the two results, we get a $\frac{1}{6}(1-\frac{1}{e})$-approximation algorithm for the election control problem in $\PLTRR$.

The next theorem generalizes~\cite[Theorem 5.2]{wilder2018controlling} as it holds for any scoring rule and for any model in which we have the ability to change only the position of $\cstar$ in the lists of a subset of voters and the increment in score of $\cstar$ is at least equal to the decrement in scoring of the other candidates.

\begin{thrm}
An $\alpha$-approximation algorithm for maximizing the increment in score of a target candidate gives an $\frac{\alpha}{3}$-approximation to the election control problem.
\label{theorem:MOVAppFactor}
\end{thrm}
\begin{proof}
Let us consider two solutions $S$ and $S^*$ for the problem of maximizing the \MOV for  candidate $\cstar$, with $S^*$ as the optimal solution to this problem.
These solutions arbitrarily select a subset of voters and modify their preference list changing the score of $\cstar$.
Let us fix $c$ and $\hat{c}$, respectively, as the candidates different from $\cstar$ with the highest score before and after the solution $S$ is applied.
Assume there exists an $\alpha$-approximation to the problem of maximizing the increment in score of the target candidate;
if we do not consider the gain given by the score lost by the most voted opponent, we have that 
\begin{align*}
    \MOV(S) &= \gplus(\cstar, S) + \gminus(\hat{c}, S) - F(\hat{c}) + F(c)    \geq \alpha \gplus(\cstar, S^*) - F(\hat{c}) + F(c)
    \\
    &\quad\geq \frac{\alpha}{3} [\gplus(\cstar, S^*) + \gminus(\bar{c}, S^*) + \gminus(\hat{c}, S^*)] - F(\hat{c}) + F(c),
\end{align*}
where the last inequality holds because $\gplus(\cstar, S) \geq \gminus(c_i, S)$ for any solution $S$ and candidate $c_i$ since $S$ modifies only the score of $\cstar$, increasing it, while the score of all the other candidates is decreased, and the increment in score to $c^*$ is equal to the sum of the decrement in score of all the other candidates.
Since $F(\hat{c})\leq F(c)$, we have that
\begin{align*}
    \MOV(S) & \geq \frac{\alpha}{3} [\gplus(\cstar, S^*) + \gminus(\bar{c}, S^*) + F(c)  + \gminus(\hat{c}, S^*) - F(\hat{c}) + F(\bar{c}) - F(\bar{c})]
    \\
    &= \frac{\alpha}{3} [\MOV(S^*) + \gminus(\hat{c}, S^*) - F(\hat{c}) + F(\bar{c})],
\end{align*}
where $\bar{c}$ is the candidate with the highest score after the solution $S^*$ is applied.
By definition of $\bar{c}$ we have that $F(\bar{c},S^*) \geq F(\hat{c},S^*)$, which implies that 
\begin{align*}
\gminus(\bar{c}, S^*) - \gminus(\hat{c}, S^*) &=
F(\bar{c})  - F(\bar{c},S^*) - F(\hat{c}) + F(\hat{c},S^*)\leq F(\bar{c}) - F(\hat{c}).
\end{align*}
Thus, $\gminus(\hat{c}, S^*) - F(\hat{c}) + F(\bar{c})\geq 0$ and we conclude that $\MOV(S) \geq \frac{\alpha}{3} \MOV(S^*)$.
\end{proof}

\paragraph{Constructive Election Control in \PLTRR.}
Next theorem shows how to get a constant factor approximation to the problem of maximizing the \MOV in $\PLTRR$ by reducing the problem to an instance of the weighted version of the influence maximization problem with LTM~\cite{kempe2015maximizing}.

This extension of the LTM, associates to each node a non-negative weight ($w: V \rightarrow \mathbb{R}^+$) 
that captures the importance of activating that node.
The goal is to find the initial seed set in order to maximize the sum of the weights of the active nodes at the end of the process, i.e., finding
\(
\argmax_{S} \sigma_w(S) = \Ex{\sum_{v \in A} w(v)},
\) where $w$ is a weight function over the node set.
\begin{algorithm}[ht]
\caption{\textsc{GreedyScore}}\label{alg:greedy}
\begin{algorithmic}[1]
\Require{Social graph $G=(V,E)$; Budget $B$}
\State $S = \emptyset$; $\hat{G} = (G, w)$  \Comment{Weighted graph $\hat{G}$}
\While{$|S| \leq B$}
\State $v = \argmax_{u \in V \setminus S} \sigma_w(S \cup \{u\}) - \sigma_w(S)$
\State $S = S \cup \{ v \}$
\EndWhile
\State \Return{$S$}
\end{algorithmic}
\end{algorithm}

A simple hill-climbing greedy algorithm achieves a $(1-1/e)$-approximation if the weights are polynomial (or polynomially small) in the number of nodes of the graph and the number of live-edge graph samples is polynomially large in the weights~\cite{kempe2015maximizing}.%
\footnote{It is still an open question how well the value of $\sigma_w(S)$ can be approximated for an influence model with arbitrary node weights.}
Intuitively, if a node has an exponentially small probability of being sampled in the live-edge graph associated with a high weight, then a polynomial number of samples would not be enough to consider it in the solution with non-negligible probability.
We exploit this result to approximate the \MOV via Algorithm~\ref{alg:greedy}, reducing the problem of maximizing the score to that of maximizing $\sigma_w(S)$ in the weighted LTM.
We define a new graph $\hat{G}$ with the same sets of nodes and edges of $G$. 
Then, we assign a weight to each node $v \in V$ equal to $w(v) := \sum_{u \in \Nout{v}}b_{vu}(1-\pi_u(\cstar))$.
Note that we are able to correctly approximate the value of $\sigma_w(S)$ using such weights since by hypothesis on the model $b_{uv} \geq \frac{1}{|V|^{\gamma_1}}$, for each $(u,v) \in E$ and for some constant $\gamma_1 > 0$, and since $\pi_v(c_i) \geq \frac{1}{|V|^{\gamma_2}}$, for each $v \in V$ for some constant $\gamma_2 > 0$.
By applying a multiplicative form of the Chernoff bound we can get a $1 \pm \epsilon$ approximation of $\sigma_w(S)$, with high probability~\cite[Proposition~4.1]{kempe2015maximizing}.

Thus, we can use Algorithm~\ref{alg:greedy} to maximize the influence on $\hat G$. 
The algorithm starts with an empty set $S$ and adds to it, in each of $B$ rounds, the node $v$ with maximal marginal gain w.r.t.\ the solution computed so far.
\begin{thrm}
\label{theorem:1-playerApproxRatio}
    Algorithm~\ref{alg:greedy} guarantees a $\frac{1}{6}(1-\frac{1}{e})$-approximation factor to constructive election control in $\PLTRR$.
\end{thrm}
\begin{proof}
\label{proof:1-playerApproxRatio}
We first prove that Algorithm~\ref{alg:greedy} gives an $\frac{1}{2}(1-\frac{1}{e})$-approximation to the problem of maximizing the increment in score of the target candidate $\cstar$ in $\PLTRR$.
Let $S$ and $S^\star$ respectively be the set of initial seed nodes found by the greedy algorithm and the optimal one.
We have that
\begin{align*}
    \gplus(\cstar,S)
    &= F(\cstar, S) - F(\cstar) \\
    &  =\sum_{v \in V} \frac{\pi_v(\cstar) + \sum_{u \in A \cap \Nin{v}}b_{uv}}{1 + \sum_{u \in A \cap \Nin{v}}b_{uv}} -\sum_{v \in V} \pi_v(\cstar)\\
    & = \sum_{v \in V} \frac{(1-\pi_v(\cstar))\sum_{u \in A \cap \Nin{v}}b_{uv}}{1 + \sum_{u \in A \cap \Nin{v}}b_{uv}}
\end{align*}
and, since the denominator is at most 2, that
\[
\gplus(\cstar,S)
\geq \frac{1}{2} \sum_{v \in V} (1-\pi_v(\cstar))\sum_{u \in A \cap \Nin{v}}b_{uv}  
= \frac{1}{2}\sum_{u \in A} \sum_{v \in \Nout{u}}b_{uv}(1-\pi_v(\cstar)),
\]
where $A$ is the set of active nodes at the end of the process.

Note that $\sum_{u \in A} \sum_{v \in \Nout{u}}b_{uv}(1-\pi_v(\cstar))$ is exactly the objective function that the greedy algorithm maximizes.
Hence, using the result by Kempe et al.~\cite{kempe2015maximizing}
we know that 
{\[
\sum_{u \in A} \sum_{v \in \Nout{u}}b_{uv}(1-\pi_v(\cstar)) \geq 
\left(1 - \frac{1}{e}\right)
\sum_{u \in A^\star}\! \sum_{v \in \Nout{u}} \! b_{uv}(1-\pi_v(\cstar)),
\]}%
where $A^\star$ is the set of active nodes at the end of the process starting from $S^\star$.

Therefore
\(
    \gplus(\cstar,S) \geq  \frac{1}{2} (1 - 1/e)\,\gplus(\cstar, S^\star)
\)
since 
{\begin{align*}
  \gplus(\cstar, S^\star) 
  &= \sum_{v \in V} \frac{(1-\pi_v(\cstar))\sum_{u \in A^\star \cap \Nin{v}}b_{uv}}{1 + \sum_{u \in A^\star \cap \Nin{v}}b_{uv}}
  \leq\sum_{u \in A^\star} \sum_{v \in \Nout{u}}b_{uv}(1-\pi_v(\cstar)),
\end{align*}}%
where the inequality holds since all the denominators in $\gplus(\cstar, S^\star)$ are at least 1.
Thus, Algorithm~\ref{alg:greedy} achieves a $\frac{1}{2}\left(1-\frac{1}{e}\right)$-approximation to the maximum increment in score.
Using Theorem~\ref{theorem:MOVAppFactor} we get a $\frac{1}{6}\left(1-\frac{1}{e}\right)$-approximation for the \MOV.
\end{proof}

\paragraph{Destructive Election Control in \PLTRR.}
The \emph{destructive election control} problem is similar to the constructive problem, but in this scenario, in our models, the probability that a voter $v$ votes for $\cstar$ \emph{decreases} depending on the amount of influence received by $v$ and the loss of probability of $\cstar$ is evenly split over all the other candidates.
In this way, we avoid negative values and values that do not sum to 1. 
In detail, if $A$ is the set of active nodes at the end of LTM, then, for each $v \in V$, the preference list $\pi_v$ changes as follows:
\[
\tilde{\pi}_v(\cstar) = \frac{\pi_v(\cstar)}{1 + \sum_{u\in A \cap \Nin{v}} b_{uv}}
\,\text{ and }\,
\tilde{\pi}_v(c_i) = \frac{\pi_v(c_i) + \frac{1}{m-1} \sum_{u\in A \cap \Nin{v}} b_{uv}}{1 + \sum_{u\in A \cap \Nin{v}} b_{uv}}
\]
for each $c_i \neq \cstar$.
We define \MOVD, i.e., what we want to maximize, as
{\begin{align*}
    \MOVD(S) &:= F(\hat c, S) - F(\cstar, S) - (F(c, \emptyset) - F(\cstar, \emptyset))
    = \gminus(\cstar, S) + \gplus(\hat c, S) + \Delta,
\end{align*}}%
where $S$ is the initial set of seed nodes and $\Delta = F(\hat c, \emptyset) - F(c, \emptyset)$ is the sum of constant terms that are not modified by the process. Note that maximizing $\MOVD$ is $\mathit{NP}$-hard (it can be proved with a similar argument to that of Theorem~\ref{theorem:np-hardness}). 

Similarly to the constructive case, we define a new graph $\hat{G}$ with the same sets of nodes and edges of $G$. 
Then, we assign a weight to each node $v \in V$ equal to $w(v) := \sum_{u \in \Nout{v} } b_{vu} \pi_u(\cstar)$ and we run Algorithm~\ref{alg:greedy} to find a seed set that approximates the maximum expected weight of active nodes.
\begin{thrm}
\label{theorem:destructive}
    Algorithm~\ref{alg:greedy} guarantees a $\frac{1}{4}(1-\frac{1}{e})$-approximation factor to the \emph{destructive election control} in $\PLTRR$.
\end{thrm}
\begin{proof}
We first prove that Algorithm~\ref{alg:greedy} achieves an $\frac{1}{2}(1-\frac{1}{e})$ approximation factor to the problem of maximizing the decrease in score of the target candidate $\cstar$ in $\PLTRR$.
Let $S$ and $S^\star$ respectively be the set of initial seed nodes found by the greedy algorithm and the optimal one.
Let $g^- _D(\cstar, S)$ be the decrease in score of candidate $\cstar$ with solution $S$, i.e., $g^- _D(\cstar, S) = F(\cstar, \emptyset) - F(\cstar, S)$.
Let $A$ be the set of active nodes at the end of the process; then we have that
{\[
g^- _D(\cstar, S) 
=\sum_{v \in V} \frac{\pi_v(\cstar)\sum_{u \in A \cap \Nin{v}} b_{uv}}{1 + \sum_{u \in A \cap \Nin{v}}b_{uv}}
\]}%
and, since the denominator is at most 2, that
{
\begin{align*}
    g^- _D(\cstar, S) & \geq \frac{1}{2} \sum_{v \in V} \left( \pi_v(\cstar)\sum_{u \in A \cap \Nin{v} }b_{uv} \right)
    \\
    & = \frac{1}{2}\sum_{u \in A} \sum_{v \in \Nout{u} } \pi_v(\cstar) \cdot b_{uv}.
\end{align*}}%

Note that \(\sum_{u \in A} \sum_{v \in \Nout{u} } \pi_v(\cstar) \cdot b_{uv}\) is exactly the objective function of the greedy Algorithm that maximizes the weighted-LTM for \(\hat{G}\).
Hence, using the result by Kempe et al.~\cite{kempe2015maximizing}, we know that 
{
\[
    \sum_{u \in A} \sum_{v \in \Nout{u} }b_{uv}\, \pi_v(\cstar)
    \geq
    \left(1-\frac{1}{e}\right)
    \sum_{u \in A^\star} \sum_{v \in \Nout{u}}b_{uv} \pi_v(\cstar),
\]}%
where $A^\star$ is the optimal set of active nodes, i.e., the set of active nodes at the end process starting from $S^\star$ ($S^\star$ the optimal solution for the weighted-LTM).

Therefore
{
\[
    g^- _D(\cstar, S) \geq \frac{1}{2}  \left(1 - \frac{1}{e}\right) g^- _D(\cstar, S^\star)
\]}%
because
{
\begin{align*}
    g^- _D(\cstar, S^\star) 
    &= \sum_{v \in V} \frac{\pi_v(\cstar)\sum_{u \in A^\star \cap \Nin{v}}b_{uv}}{1 + \sum_{u \in A^\star \cap \Nin{v}}b_{uv}}
    \\
    &\leq \sum_{v \in V} \pi_v(\cstar) \sum_{u \in A^\star \cap \Nin{v}}b_{uv}
    = \sum_{u \in A^\star} \sum_{v \in \Nout{u}}b_{uv}\, \pi_v(\cstar),
\end{align*}}%
where the inequality is due to the fact that the denominator in all the terms of $g^- _D(\cstar, S^\star)$ is at least 1.
Thus we achieve a $\frac{1}{2}\left(1-\frac{1}{e}\right)$-approximation to the maximum increment in score.

Let us fix $c$ and $\hat{c}$, respectively, as the candidates different from $\cstar$ with the highest score before and after the solution $S$ is applied; let $\bar{c}$ be the most voted opponent after the optimal solution $S^*$ is applied. 
Then we have that
{
\begin{align*}
    &\MOV(S) = g^- _D(\cstar, S) + g^+ _D(\hat c, S)  + F(\hat c, \emptyset)- F(c, \emptyset)
    \\
    &\geq 
    \frac{1}{2} \left(1-\frac{1}{e}\right) g^- _D(\cstar, S*) + g^+ _D(\hat c, S)  + F(\hat c, \emptyset)- F(c, \emptyset)
    \\
    &\geq 
    \frac{1}{4} \left(1-\frac{1}{e}\right) \left[
    g^- _D(\cstar, S^*) + g^+ _D(\bar c, S^*)  + g^+ _D(\hat c, S) + F(\hat c, \emptyset)- F(c, \emptyset)
    \right]
    \\
    &=
    \frac{1}{4} \left(1-\frac{1}{e}\right) \left[ \MOV(S^*) + g^+ _D(\hat c, S)  + F(\hat c, \emptyset) - F(\bar c, \emptyset)
    \right]
    \\
    &\geq 
    \frac{1}{4} \left(1-\frac{1}{e}\right) \MOV(S^*),
\end{align*}}%
where the last inequality holds since, by definition of $\bar c$ and $\hat c$, we have that
\(
	 g^+ _D(\hat c, S)  + F(\hat c, \emptyset) \geq g^+ _D(\bar c, S)  + F(\bar c, \emptyset).
\)
\end{proof}
\section{Simulations}
\label{sec:simulation}

\begin{figure}[ht!]
\begin{center}
\begin{minipage}[t]{0.5\textwidth}
  \includegraphics[width=\linewidth]{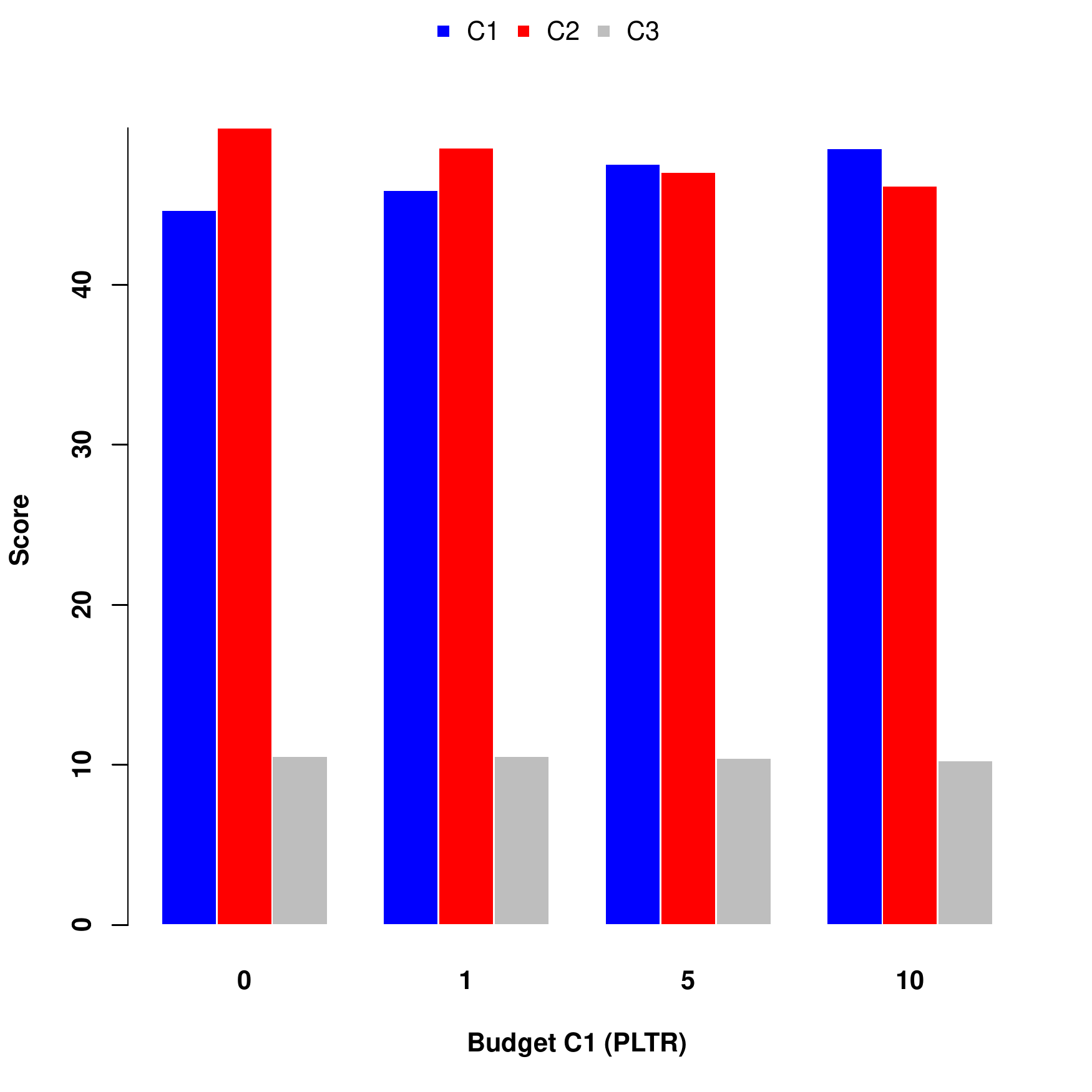}
\end{minipage}\hfill%
\begin{minipage}[t]{0.5\textwidth}
  \includegraphics[width=\linewidth]{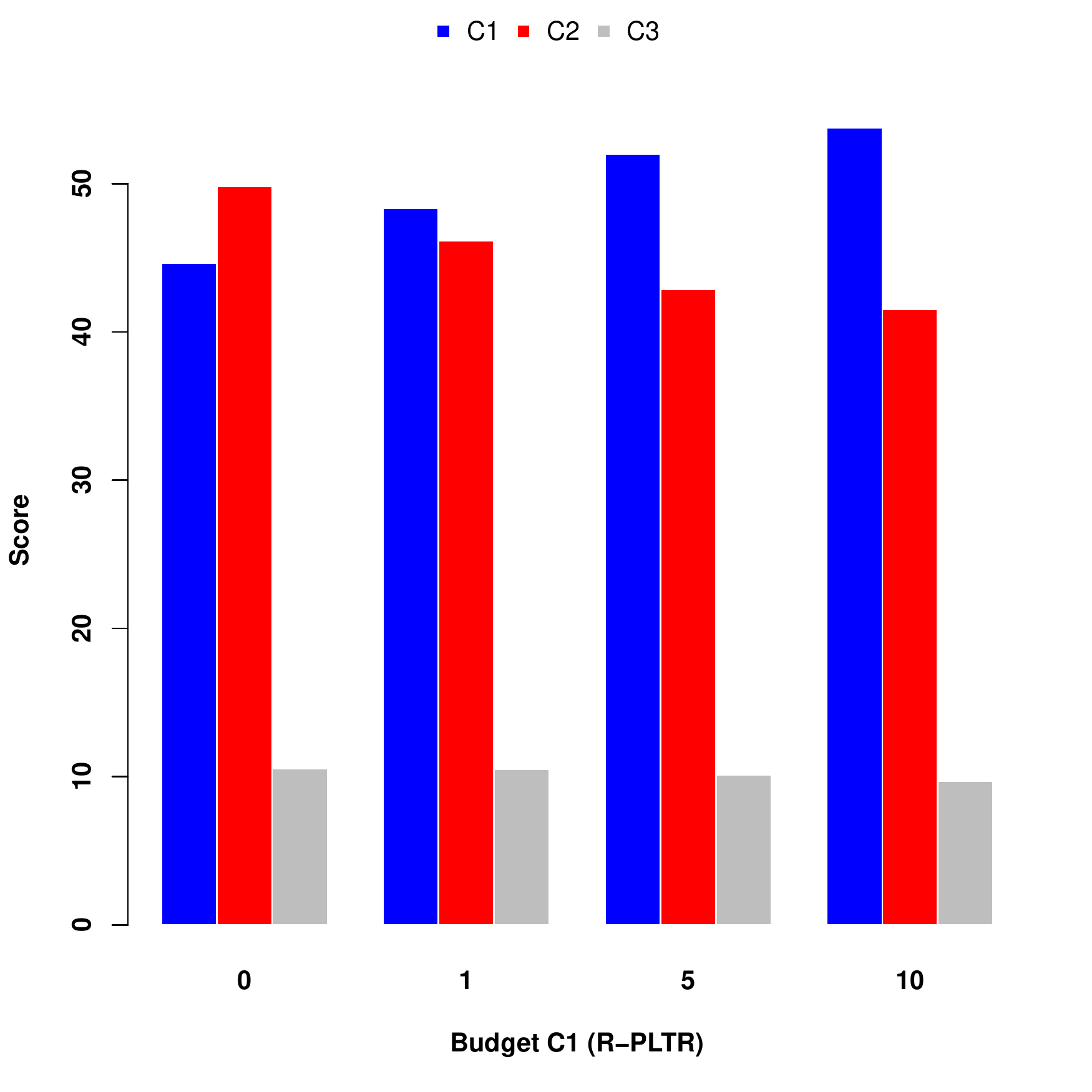}
\end{minipage}\hfill%
\begin{minipage}[t]{0.5\textwidth}
  \includegraphics[width=\linewidth]{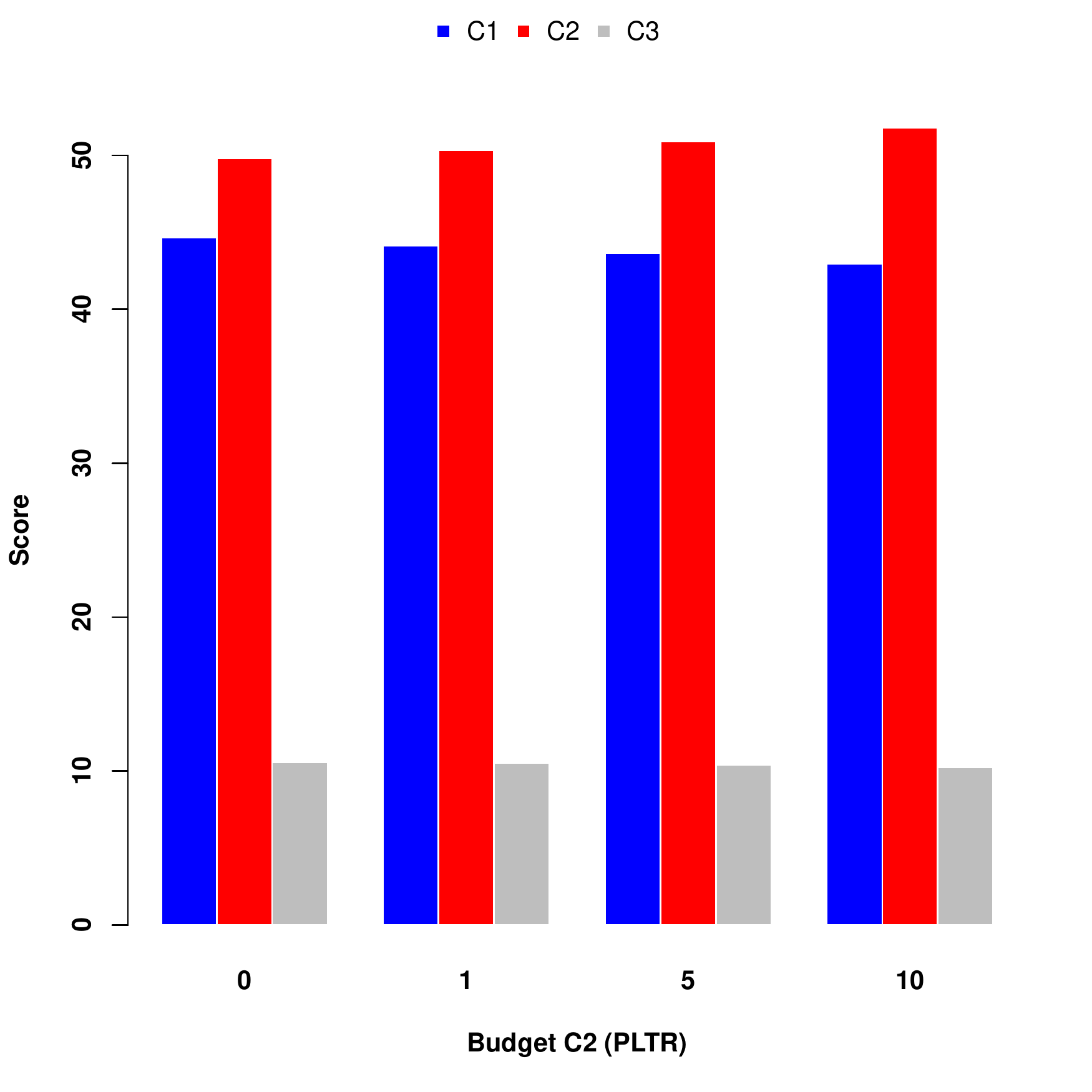}
\end{minipage}\hfill%
\begin{minipage}[t]{0.5\textwidth}
  \includegraphics[width=\linewidth]{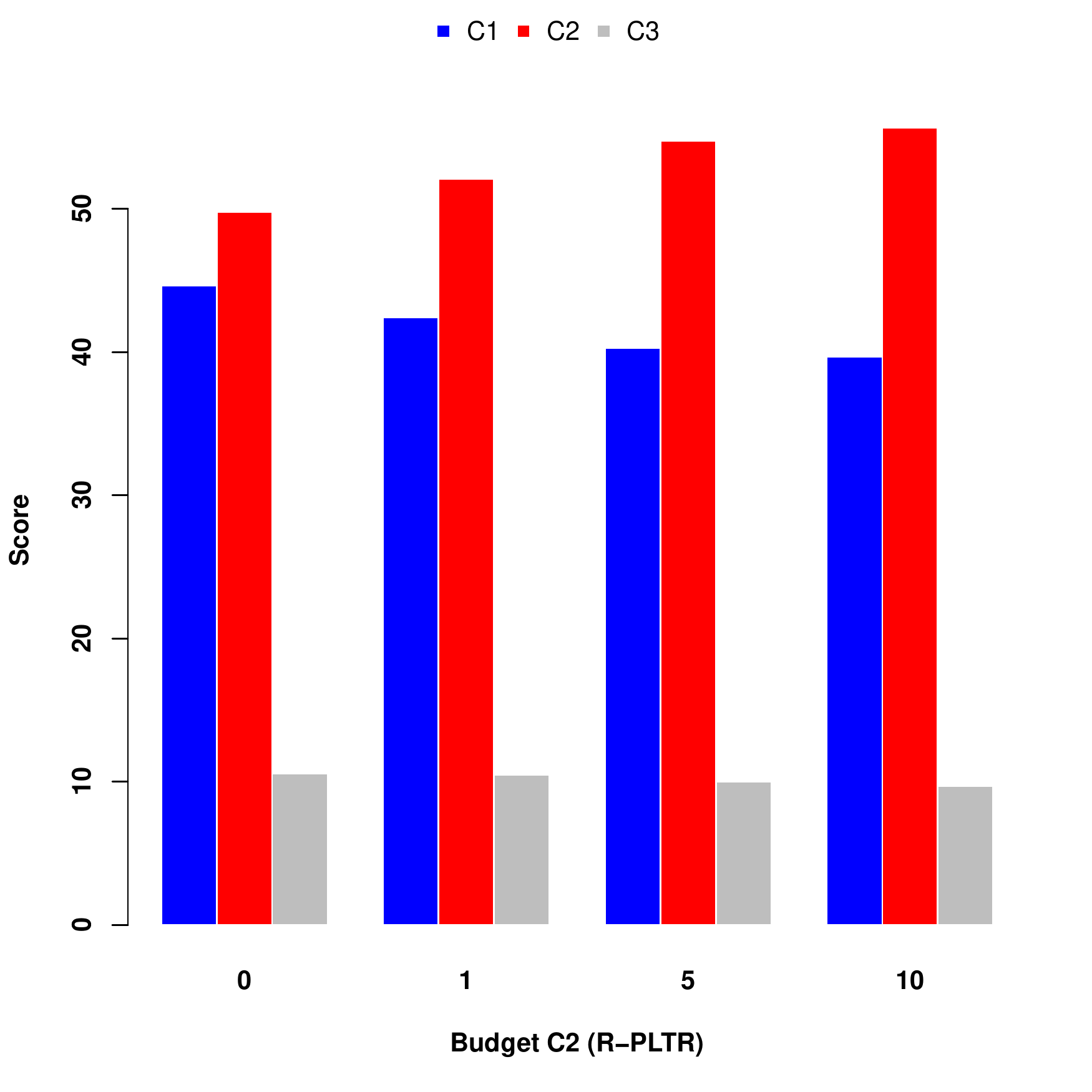}
\end{minipage}\hfill%
\begin{minipage}[t]{0.5\textwidth}
  \includegraphics[width=\linewidth]{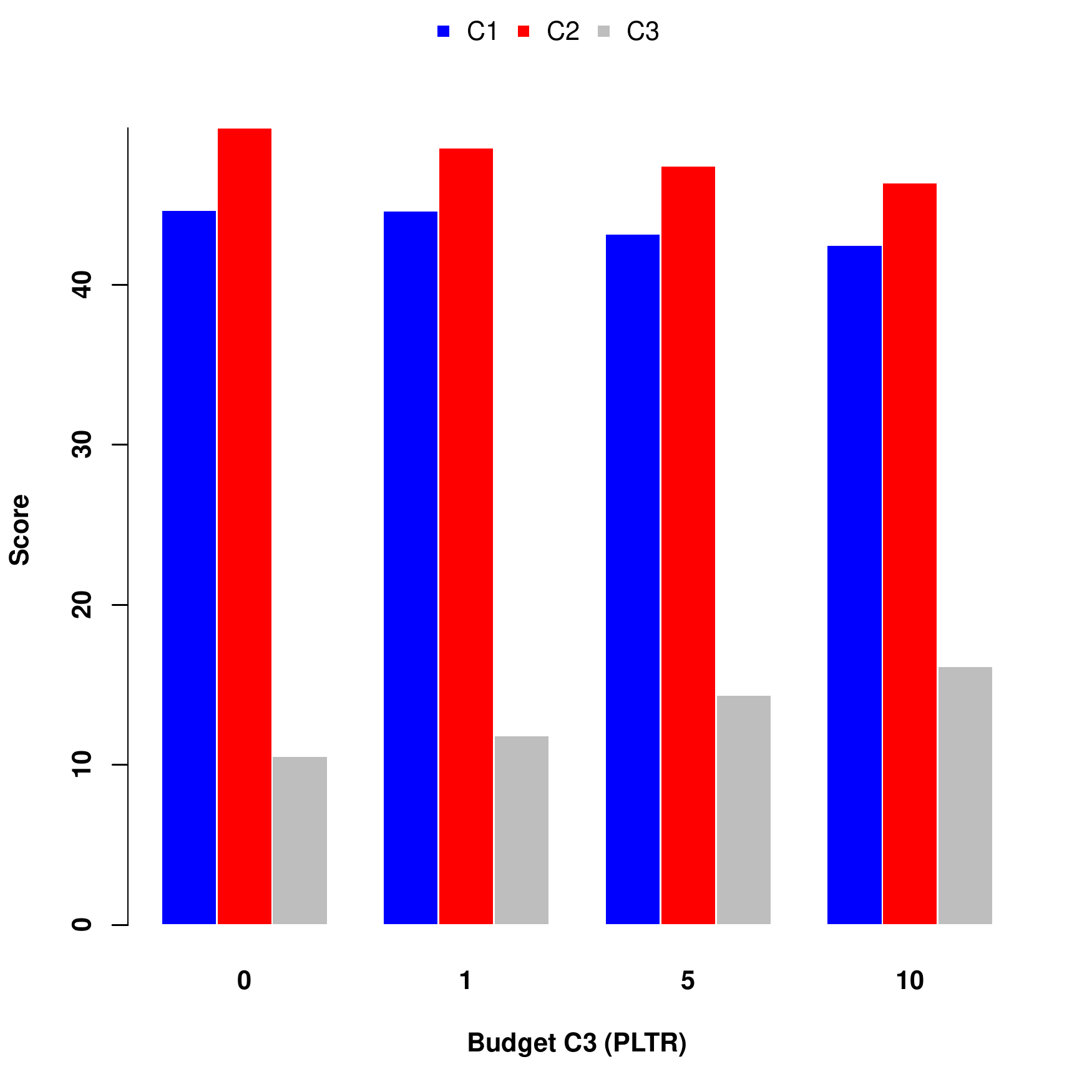}
\end{minipage}\hfill%
\begin{minipage}[t]{0.5\textwidth}
  \includegraphics[width=\linewidth]{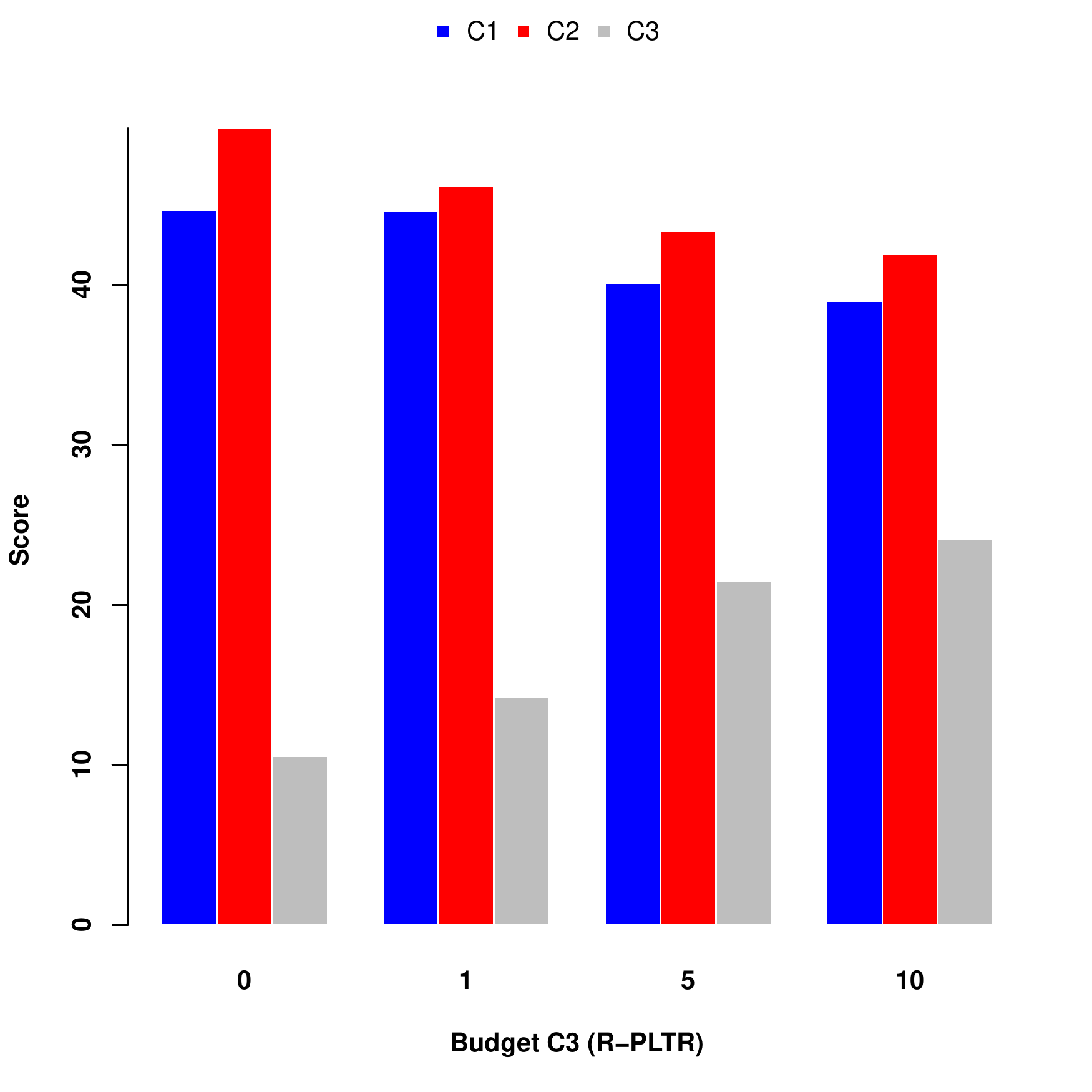}
\end{minipage}
\caption{
Candidates' scores in \emph{polbooks} in $\PLTR$ (left column) and $\PLTRR$ (right column), considering as target candidate the ``liberal'' (top), the ``conservative'' (center), and the ``neutral'' (bottom).
}
\label{fig:scores_polbooks}
\end{center}
\end{figure}

\begin{figure}[ht!]
\begin{center}
    \begin{minipage}[t]{0.5\textwidth}
      \includegraphics[width=\linewidth]{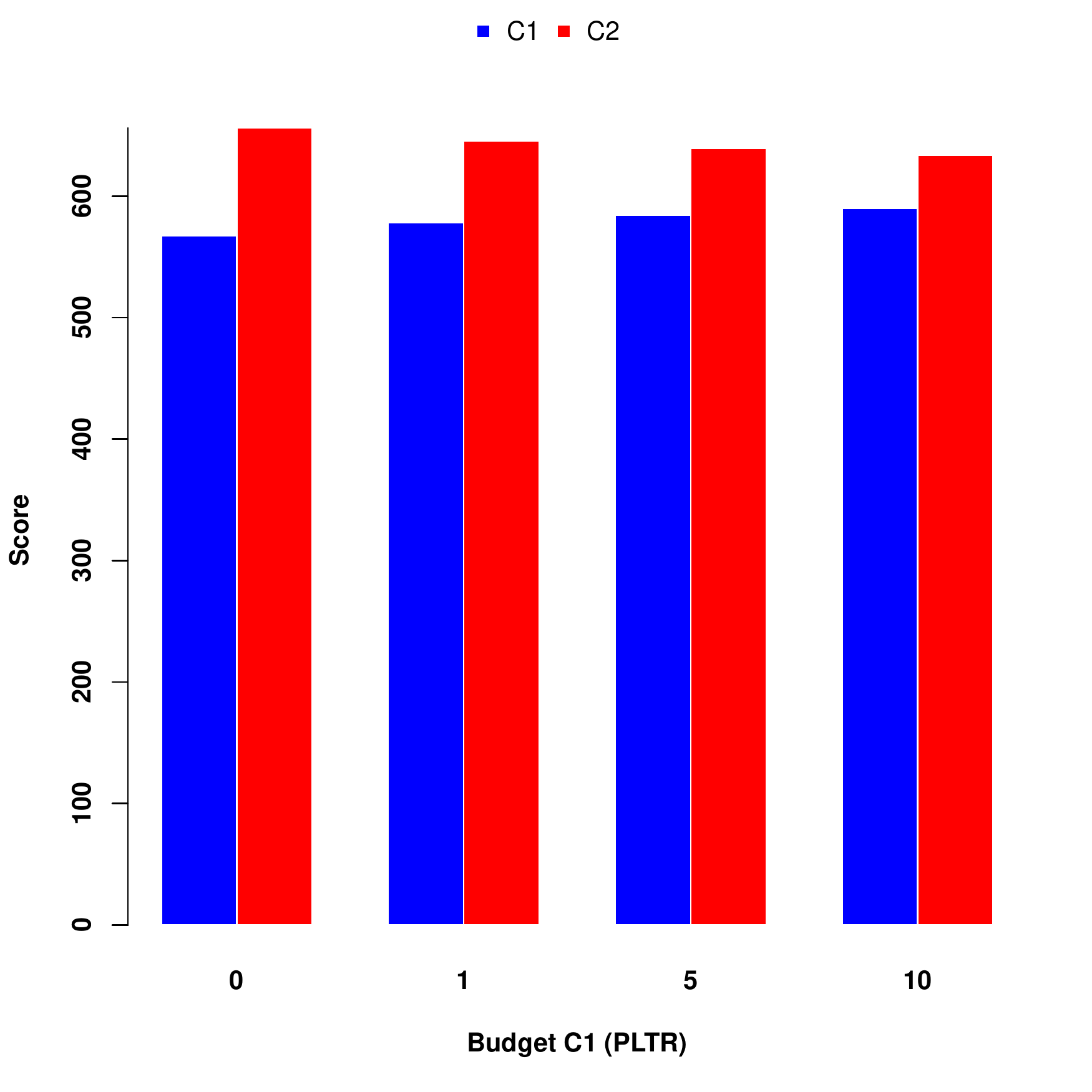}
    \end{minipage}\hfill%
    \begin{minipage}[t]{0.5\textwidth}
      \includegraphics[width=\linewidth]{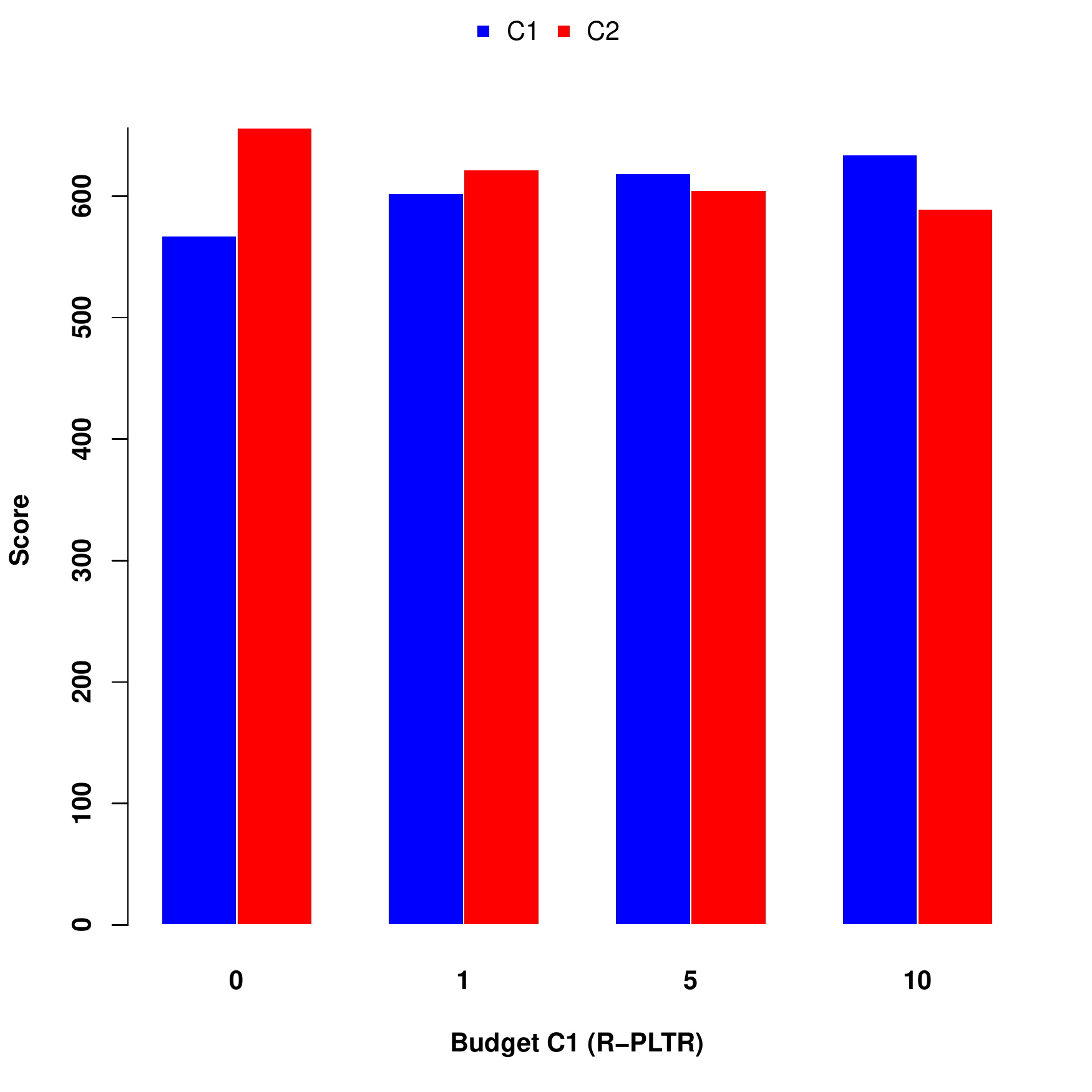}
    \end{minipage}\hfill%
    \begin{minipage}[t]{0.5\textwidth}
      \includegraphics[width=\linewidth]{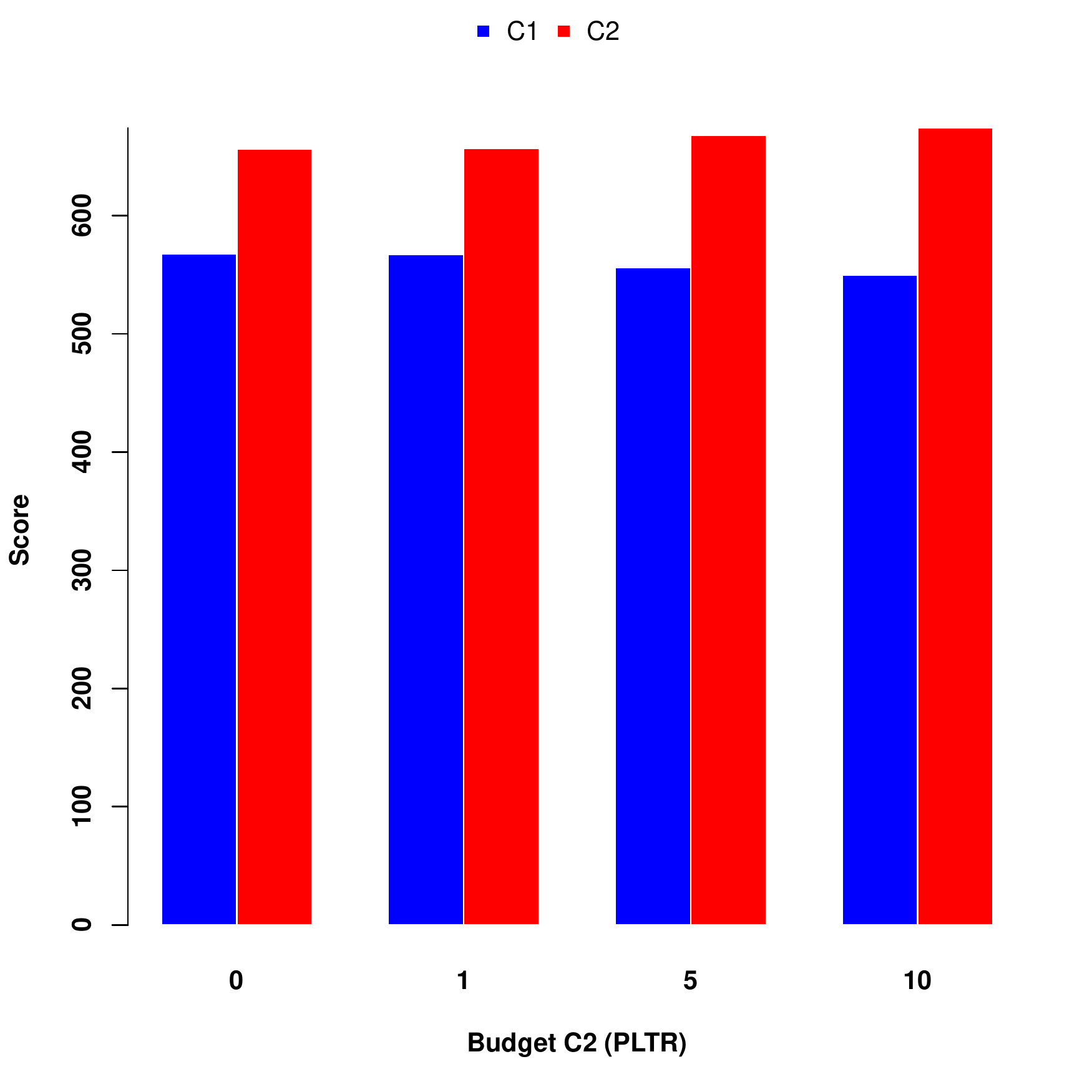}
    \end{minipage}\hfill%
    \begin{minipage}[t]{0.5\textwidth}
      \includegraphics[width=\linewidth]{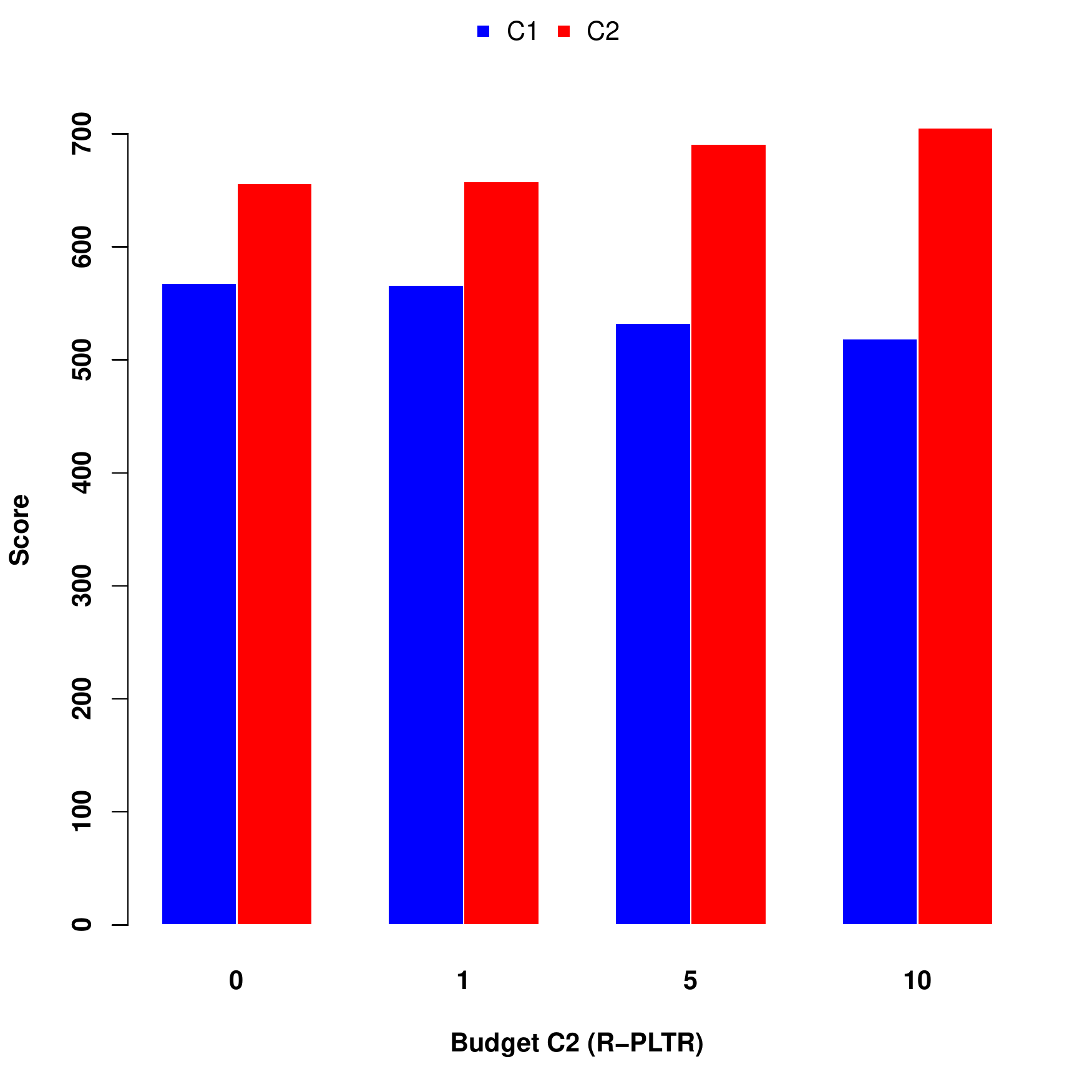}
    \end{minipage}
\end{center}
\caption{
Candidates' scores in \emph{polblogs} in $\PLTR$ (left column) and $\PLTRR$ (right column), considering as target candidate the ``liberal'' (top) and the ``conservative'' (bottom).
}
\label{fig:scores_polblogs}
\end{figure}

\begin{figure}[ht!]
\begin{minipage}[t]{0.5\textwidth}
  \includegraphics[width=\linewidth]{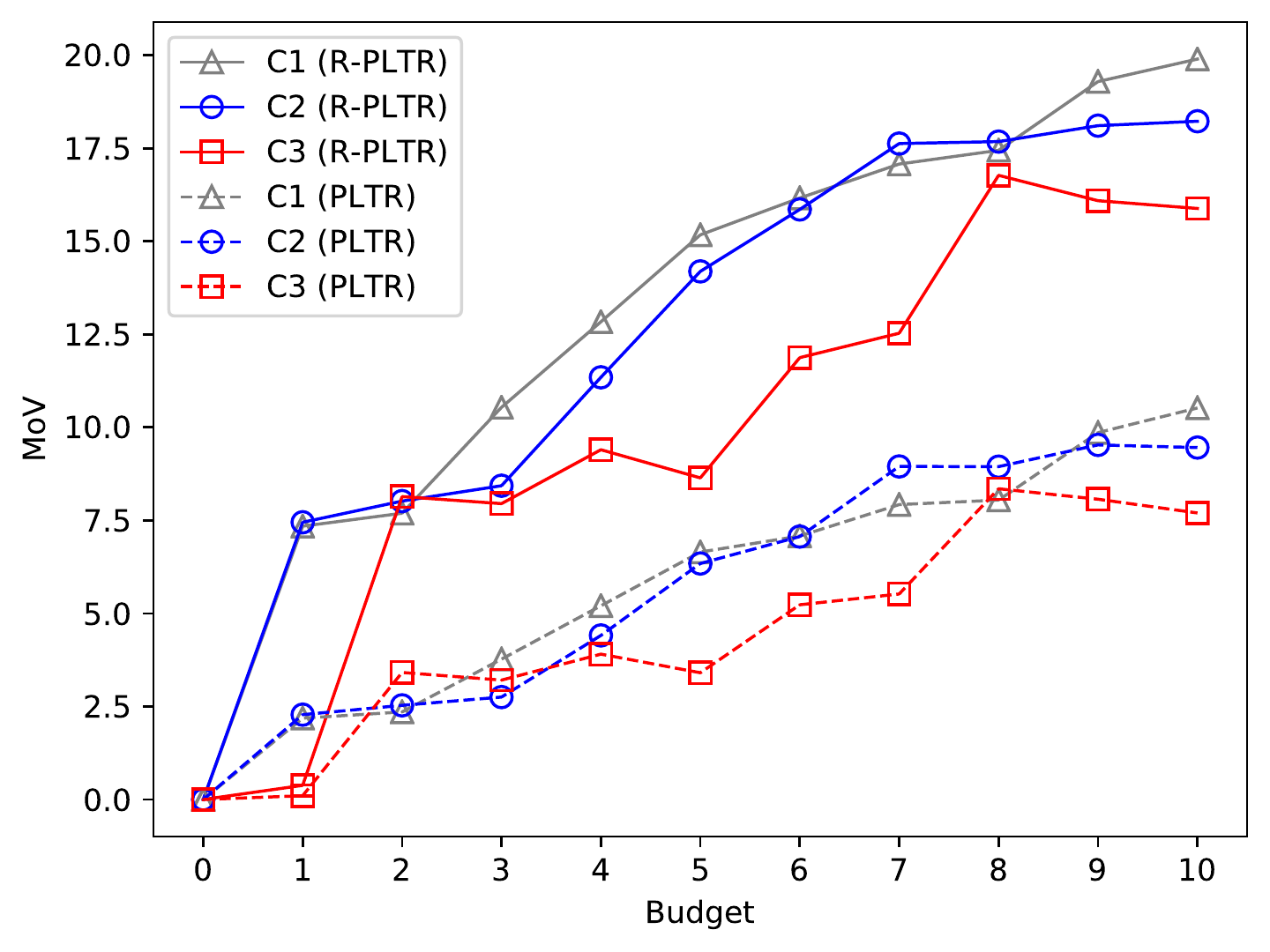}
\end{minipage}\hfill%
\begin{minipage}[t]{0.5\textwidth}
  \includegraphics[width=\linewidth]{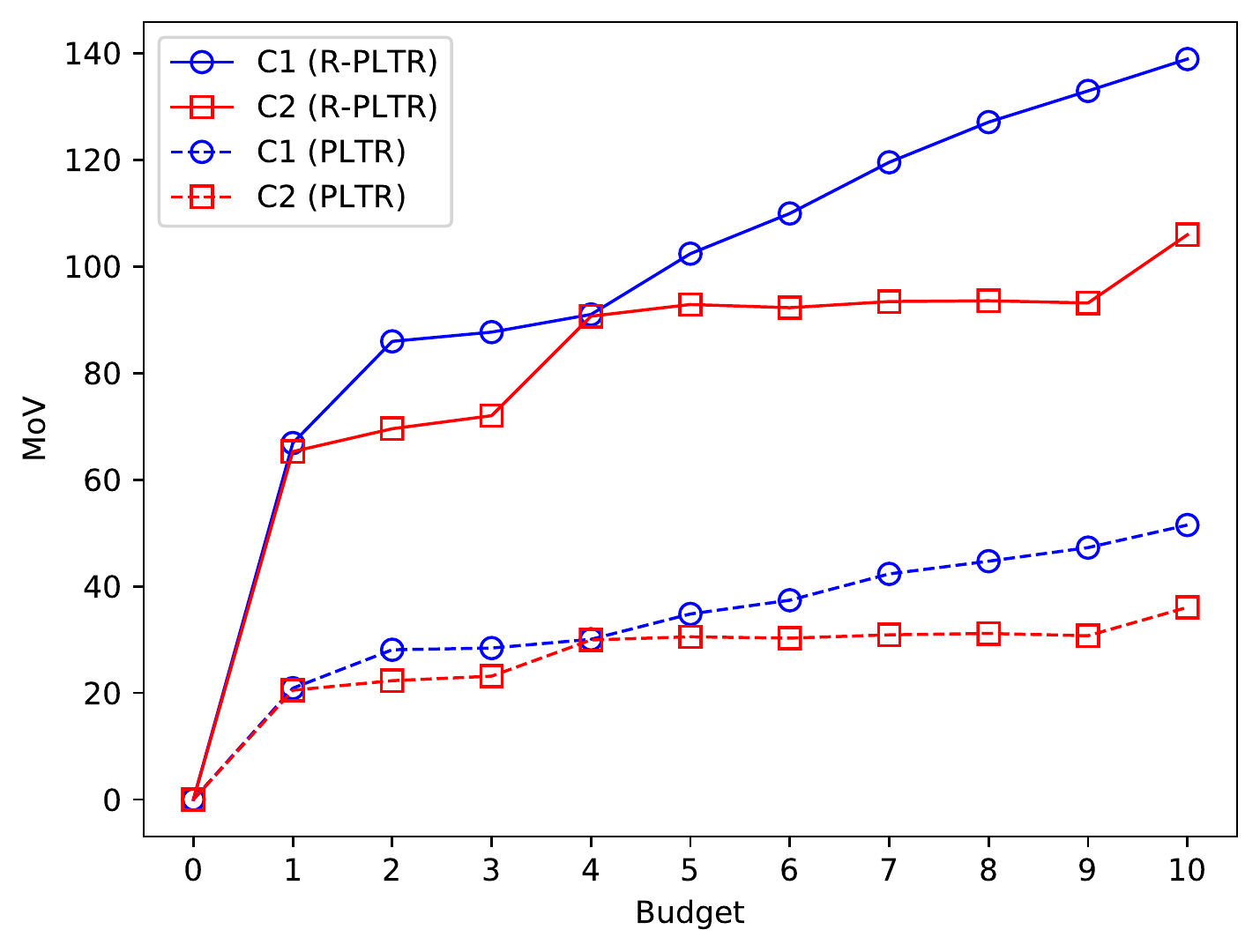}
\end{minipage}
\caption{The $\MOV$ calculated using the presented algorithm for \emph{polbooks} (left) and \emph{polblogs} (right), both in $\PLTR$ (dashed line) and $\PLTRR$ (solid line), considering as target candidate the ``liberal'' (blue line), the ``conservative'' (red line), and the ``neutral'' (grey line).}
\label{fig:mov}

\begin{minipage}[t]{0.5\textwidth}
  \includegraphics[width=\linewidth]{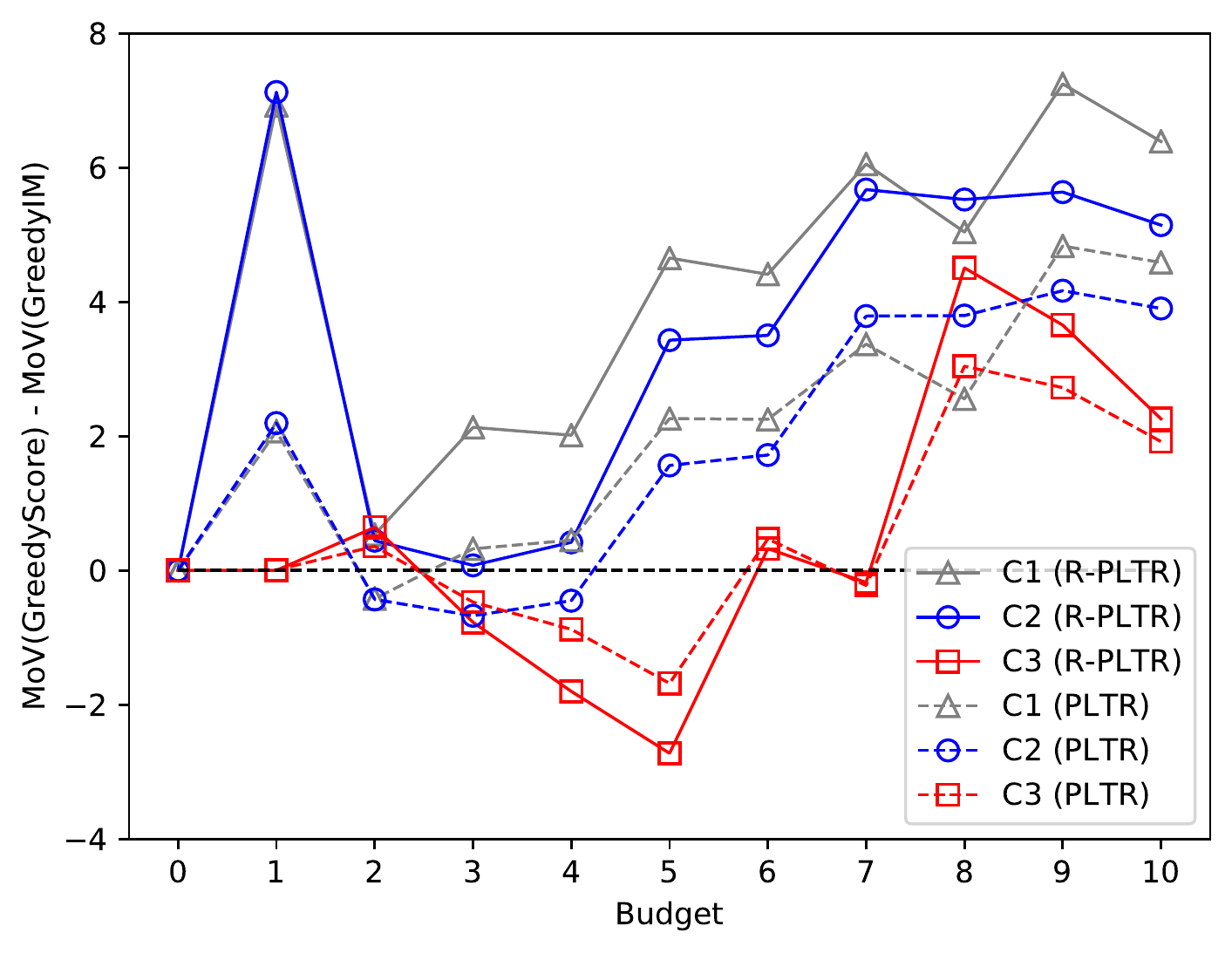}
\end{minipage}\hfill%
\begin{minipage}[t]{0.5\textwidth}
  \includegraphics[width=\linewidth]{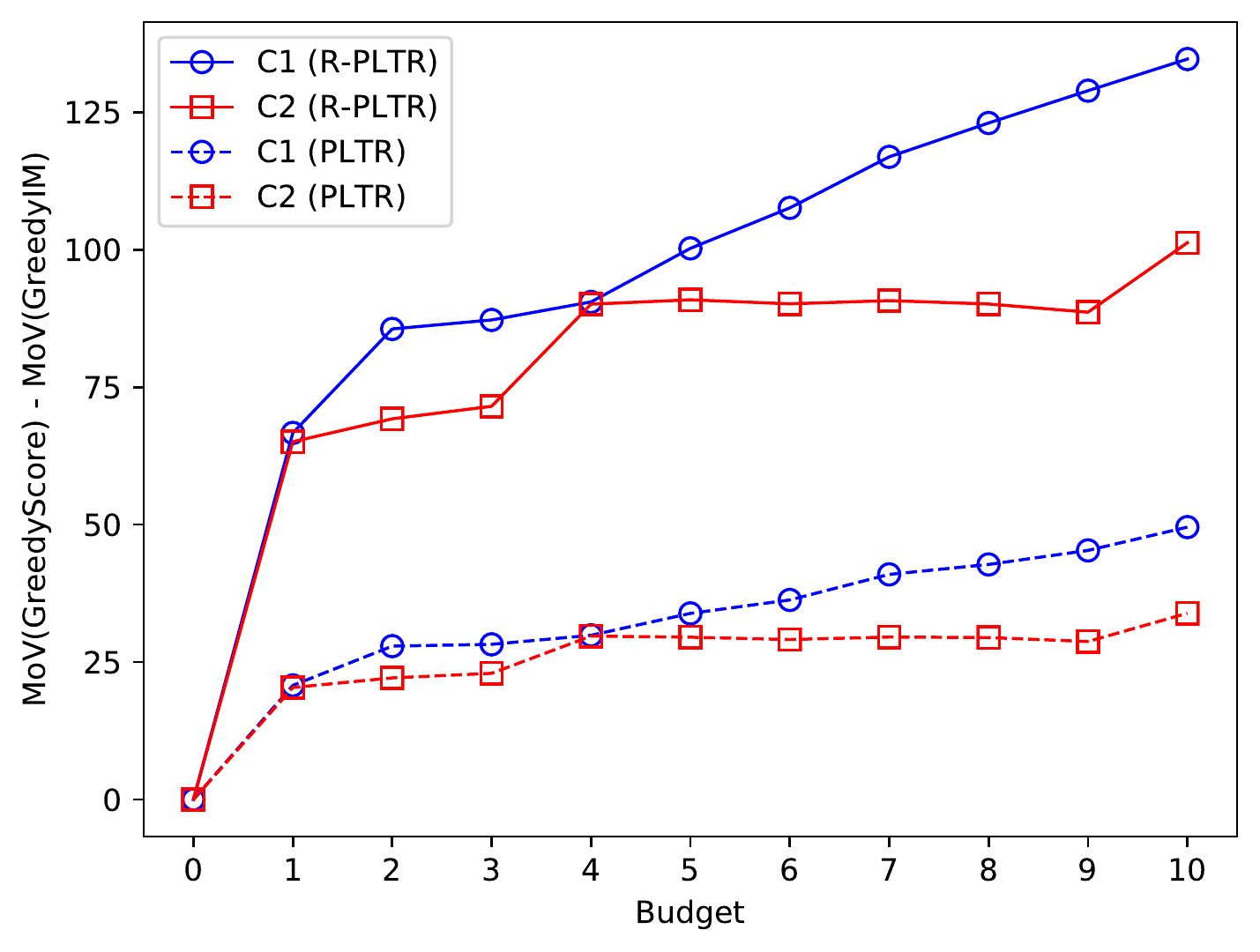}
\end{minipage}
\caption{Difference between $\MOV$ obtained using our greedy algorithm and $\MOV$ obtained using the standard greedy algorithm for Influence Maximization problem in \emph{polbooks} (left) and \emph{polblogs} (right). 
Values greater than 0 are when our algorithm performs better than the simple Greedy for Influence Maximization.
}
\label{fig:greedy}
\end{figure}

We simulate our model on two real-world social networks%
\footnote{The datasets are taken from \url{http://networkrepository.com/}}
on which political campaigning messages could spread:
\begin{itemize}
    \item \emph{polbooks}: an undirected network with 105 nodes and 882 edges where nodes are political books and edges represent co-purchasing behavior; nodes are labeled as ``liberal,'' ``conservative,'' or ``neutral.''
    \item \emph{polblogs}: a directed network with 1,224 nodes and 19,025 edges where nodes are web blogs about US politics and edges hyperlinks connecting them; nodes are labeled as ``liberal'' or ``conservative.''
\end{itemize}
The number of candidates in our simulations is based on the ground truth of the datasets; 
as mentioned earlier, \emph{polbooks} has three clusters and \emph{polblogs} has two clusters based on different US political parties.
We set the probability of each node $v$ to vote for, say, a ``liberal'' candidate proportionally to the number of neighbors labeled as ``liberal,'' 
i.e., we set $\pi_v(c) = \frac{|N_v \cap B|}{|N_v|}$ where $c$ is the ``liberal'' candidate, $B$ is the set of nodes labeled as ``liberal,'' and $N_v$ is the set of neighbors of $v$.
For each node $v$ we sampled the ``non-incoming influence weight'' $\bar{b}_v$ uniformly at random in $[0,1]$ and assigned the remaining influence weight uniformly among its incoming neighbors, i.e., we assigned to each edge $(u,v)$ a weight $b_{uv} = \frac{1-\bar{b}_v}{|N_v^i|}$.

In our simulations, we run \textit{GreedyScore} (Algorithm~\ref{alg:greedy}) for the election control problem in $\PLTRR$.
Then, we measure the score and the \MOV of each candidate using as starting seed nodes the ones found by the algorithm both in $\PLTR$ and in $\PLTRR$.
We run the simulation considering each different candidate as the target one to cover multiple scenarios, considering as budget values the ones in $\{0,1,5,10\}$.
Then, as baseline to compare, we also considered as seed nodes the most influential ones, i.e., the nodes selected by \textit{GreedyIM}, the classical greedy algorithm for Influence Maximization~\cite{kempe2015maximizing}.

For the implementation, we used .Net framework 4.6.2 and C\# programming language. 
We have implemented five different classes for managing the graph, the LTM process, the $\PLTR$ process, and a GUI.
We execute the simulations on a system with the following specifications: CPU Intel Core i7-6700HQ 2.6 GHz, with $4 \times 32$ KB 8-way L1 (data and inst) cache, and $4 \times 256$ KB 4-way L2 cache, and $6$ MB 12-way L3 cache, RAM 16G DDR4.
Each simulation has a running time of approximately 40 seconds for \emph{poolbooks} and 140 minutes for \emph{polblogs}. 

The results relative to the scores are shown in Figures~\ref{fig:scores_polbooks} and~\ref{fig:scores_polblogs}.
As expected, the effect of our algorithm in $\PLTRR$ is amplified compared to $\PLTR$, since it affects a greater number of voters.
Taking as example the ``liberal'' candidate in \emph{polbooks}, we need a budget $B=5$ to make it overtake the ``conservative'' candidate in $\PLTR$, while a budget $B=1$ is enough in $\PLTRR$ (Figure~\ref{fig:scores_polbooks});
in \emph{polblogs}, instead, we are not able to make the ``liberal'' candidate win in $\PLTR$ with budget $B=10$, but it is enough a budget $B=5$ to make it overtake the ``conservative'' candidate in $\PLTRR$ (Figure~\ref{fig:scores_polblogs}).

The results relative to \MOV are presented in Figure~\ref{fig:mov}.
We can note that, as a general trend, candidates with lower probability of winning, are the most affected by the influence generated by the seed nodes selected by our algorithm both in $\PLTR$ and $\PLTRR$.
The ``neutral'' and ``liberal'' candidates, respectively last and second last voted, have the higher \MOV in \emph{polbooks} (see Figure~\ref{fig:mov}, on the left), while the ``liberal'' candidate, which was losing the elections, has the higher \MOV in \emph{polblogs} (see Figure~\ref{fig:mov}, on the right).

Finally, in Figure~\ref{fig:greedy} we present the difference between the \MOV calculated by \textit{GreedyScore} and \textit{GreedyIM}.
The simulations show that our algorithm outperforms \textit{GreedyIM}, as expected. 
The only scenario in which our algorithm performs worse is that in which we influence, with low budget, the already winning candidate (see Figure~\ref{fig:greedy}, on the left, red lines).
The reason why \textit{GreedyScore} works better than \textit{GreedyIM} is that it looks for seeds that will influence ``critical'' voters, i.e., voters on which the influence will have more impact on the global score of the candidates, while \textit{GreedyIM} just looks for influential voters, independently from their initial opinion.

\section{Conclusions and Future Work}
\label{sec:conclusion}
Influencing elections by means of social networks is a significant issue in modern society, and understanding this phenomenon is of crucial importance in order to prevent the integrity of democracy.
Our results constitute the first step towards realistic modeling of the use of social influence to control elections as our models take into account that voters might hide their preferences to a controller. 
In one of our models the election control problem cannot be approximated within any reasonable bound, under some computational complexity hypothesis. For the other model we provide an approximation algorithm that guarantees a constant factor approximation ratio.
The results in this paper open several research directions.
We plan to study the election control problem in a variant of \PLTR where multiple campaigns affect voters' opinions on different candidates. 
It is also worth to investigate models with uncertainty in other voting systems.
Finally, it would be interesting to consider uncertainty models also for the diffusion process, e.g., in robust influence maximization only a probability distribution on the edge's weights is known.

\cleardoublepage
\bibliographystyle{alpha}
\bibliography{references}

\end{document}